\documentclass[11pt]{article}
\usepackage{amsfonts,amsmath,amssymb,epsfig,footmisc,here,latexsym,setspace,verbatim}
\usepackage[all]{xy}
\usepackage[english]{babel}
\usepackage{hyperref}

\newcommand{\BlackBoxes}{\global\overfullrule5pt}

\BlackBoxes

\def\varepsilon{{\varepsilon}}

\newcommand{\be}{\begin{equation}}
\newcommand{\ee}{\end{equation}}
\newcommand{\bea}{\begin{eqnarray}}
\newcommand{\eea}{\end{eqnarray}}
\newcommand{\beas}{\begin{eqnarray*}}
\newcommand{\eeas}{\end{eqnarray*}}

\newtheorem{theorem}{Theorem}[section]

\newtheorem{proposition}[theorem]{Proposition}
\newtheorem{corollary}[theorem]{Corollary}
\newtheorem{lemma}[theorem]{Lemma}
\newtheorem{remark}[theorem]{Remark}
\newtheorem{example}[theorem]{Example}
\newtheorem{examples}[theorem]{Examples}

\newtheorem{foo}[theorem]{Remarks}

\newenvironment{proof}{\addvspace{\medskipamount}\par\noindent{\it Proof}.}
{\unskip\nobreak\hfill$\Box$\par\addvspace{\medskipamount}}

\newcommand{\E}[1]{{\rm E}\left[#1\right]}

\newcommand{\EQ}[1]{{\rm E}_{Q}\left[#1\right]}
\newcommand{\EP}[1]{{\rm E}_{P}\left[#1\right]}

\newcommand{\op}[1]{\operatorname{#1}}

\newcommand{\R}{{\mathbb R}}
\newcommand{\Nat}{{\mathbb N}}

\newcommand{\Pcal}{{\mathcal P}}

\newcommand{\Ucal}{{\mathcal U}}
\newcommand{\Vcal}{{\mathcal V}}
\newcommand{\Ccal}{{\mathcal C}}

\newcommand{\Wcal}{{\mathcal W}}

\DeclareMathOperator{\esssup}{ess\,sup}
\DeclareMathOperator{\essinf}{ess\,inf}

\DeclareMathOperator{\dom}{dom}

\begin{document}
\title{Robust Optimal Risk Sharing and Risk Premia in Expanding Pools\\[3mm]
\large Running Title: Robust Optimal Risk Sharing and Risk Premia}
\author{Thomas Knispel
\\
{\footnotesize Leibniz University of Hannover, Talanx AG}\\
{\footnotesize and Competence Center for Risk and Insurance}\\
{\footnotesize \textit{Mail:} Leibniz University of Hannover}\\
{\footnotesize Institute of Probability and Statistics}\\
{\footnotesize Welfengarten 1}\\
{\footnotesize D-30167 Hannover, Germany}\\
{\footnotesize \textit{Email:} {\tt knispel@leibniz-lab.de}}\\
\and Roger J. A. Laeven$^{*}$\\
{\footnotesize University of Amsterdam, EURANDOM}\\
{\footnotesize and CentER}\\
{\footnotesize \textit{Mail:} University of Amsterdam}\\
{\footnotesize Amsterdam School of Economics}\\
{\footnotesize PO Box 15867}\\
{\footnotesize 1001 NJ Amsterdam, The Netherlands}\\
{\footnotesize \textit{Email:} {\tt R.J.A.Laeven@uva.nl}}\\
{\footnotesize $^{*}$Corresponding author}\\[1mm]
\and Gregor Svindland\\
{\footnotesize Ludwig-Maximilians University of Munich}\\
{\footnotesize \textit{Mail:} Ludwig-Maximilians University of Munich}\\
{\footnotesize Mathematics Institute}\\
{\footnotesize Theresienstr. 39}\\
{\footnotesize D-80333 Munich, Germany}\\
{\footnotesize \textit{Email:} {\tt svindla@math.lmu.de}}\\[3mm]}
\date{\today}
\maketitle
\begin{abstract}
We consider the problem of optimal risk sharing in a pool of cooperative agents. 
We analyze the asymptotic behavior of the certainty equivalents and risk premia
associated with the Pareto optimal 
risk sharing contract as the pool expands. 
We first study this problem under expected utility preferences with an objectively or subjectively given probabilistic model.
Next, we develop a robust approach by explicitly taking uncertainty about the probabilistic model (ambiguity) into account.
The resulting robust certainty equivalents and risk premia compound risk and ambiguity aversion.
We provide explicit results on their limits and rates of convergence, 
induced by Pareto optimal risk sharing in expanding pools.
\noindent
\\[3mm]\noindent\textbf{Keywords:} Ambiguity; Convex risk measures; Large pools; Pareto optimality; Risk premia; Risk sharing; Robust preferences.
\\[3mm]\noindent\textbf{AMS 2010 Classification:} Primary: 91B06, 91B16, 91B30; Secondary: 60E15,
62P05.
\\[3mm]\noindent\textbf{JEL Classification:} D81, G10, G20.
\end{abstract}

\makeatletter
\makeatother
\maketitle

\onehalfspacing

\section{Introduction}

Risk sharing constitutes a main principle in economic and mathematical risk theory.
It refers to a subdivision of the aggregate risk in a pool by exchanging and relocating risks among the cooperative individuals that participate in the pool. Risk sharing provides a means of inducing risk reduction for the individuals, in a potentially Pareto optimal sense.
Since the seminal work by Borch \cite{B62} it has been studied by numerous authors in a wide variety of settings;
see e.g., Arrow \cite{A63}, Wilson \cite{W68}, DuMouchel \cite{D68}, Gerber \cite{G78,G79}, B\"uhlmann and Jewell \cite{BJ79}, Landsberger and Meilijson \cite{LM94},
and, more recently, Carlier and Dana \cite{CD03}, Heath and Ku \cite{HK04}, Barrieu and El Karoui \cite{BeK05,BeK09},
Dana and Scarsini \cite{DS07}, Jouini, Schachermayer and Touzi \cite{JST08},
Kiesel and R\"uschendorf \cite{KR08}, Ludkovski and R\"uschendorf \cite{LR08},
Filipovi\'c and Svindland \cite{FS08}, Dana \cite{D11},
Ravanelli and Svindland \cite{RS14}, and the references therein.

This paper explores what happens when Pareto optimal risk sharing is combined with an \textit{expanding} pool of risks.
In an expanding pool of independent and identically distributed (i.i.d.) risks,
the distribution of the aggregate risk spreads,
but the average risk obeys the law of large numbers and converges to its expectation
(see e.g., Samuelson \cite{S63}, Diamond \cite{D84} and Ross \cite{R99} for a detailed discussion).
We analyze when
the individuals' risk reduction induced by Pareto optimal risk sharing may be exploited to the full limit:
when, upon subdividing and relocating the aggregate risk according to the Pareto optimal 
risk sharing rule
in an expanding pool of i.i.d.\ risks
with cooperating individuals that have identical preferences,
will risk sharing eventually lead to annihilating risk beyond its expectation?

We answer this question by analyzing, in a general setting, the asymptotic behavior of the certainty equivalents and risk premia
in an expanding pool of risks under Pareto optimal risk sharing.
Adopting the classical expected utility model of Von Neumann and Morgenstern \cite{vNM44},
Pratt \cite{P64} studies the connection between the risk premium,
defined as the expected value of a given risk minus its certainty equivalent
(the monetary amount that makes an agent indifferent to the risk),
and the utility function.
He shows that greater local risk aversion (risk aversion in the small) at all wealth levels
implies greater global risk aversion (risk aversion in the large) and vice versa,
in the sense that the risk premia 
vary with the local risk aversion intensity. 
Furthermore, Pratt \cite{P64} provides an expansion of the risk premium
for a small and actuarially fair risk, given by the local risk aversion times half the variance of the risk.
Hence, with vanishing variance of the average risk in an expanding pool of i.i.d.\ risks---as implied by the law of large numbers---,
the risk premium 
associated with the Pareto optimal risk sharing rule can be seen to vanish, too.
We analyze this convergence formally and derive results on the risk premium's rate of convergence.
We first consider the relatively simple case of the expected utility model, as in Pratt \cite{P64} but with refined results,
and next turn to more advanced decision models, for which the problem proves to be much more delicate.

In recent years, the distinction between risk (probabilities given) and ambiguity (probabilities unknown) has received much attention.
Under Savage's \cite{S54} subjective expected utility model this distinction is absent
due to the assignment of subjective probabilities.
Modeling approaches that explicitly recognize the fact that a specific probabilistic model may be misspecified are referred to as robust
(see e.g., Hansen and Sargent \cite{HS01,HS07}).
A popular class of models for decision under risk and ambiguity is provided by the multiple priors models (Gilboa and Schmeidler \cite{GS89};
see also Schmeidler \cite{S86,S89}).
It occurs as a special case of the rich variational and homothetic preference models (Maccheroni, Marinacci and Rustichini \cite{MMR06},
Cerreia-Vioglio et al. \cite{CMMM08} and Chateauneuf and Faro \cite{CF10}).
These models all reduce to the expected utility model of Von Neumann and Morgenstern \cite{vNM44}
when ambiguity has resolved in the classical Anscombe and Aumann \cite{AA63} setup.
A related strand of literature in financial mathematics is that of convex measures of risk introduced by
F\"ollmer and Schied \cite{FS02}, Frittelli and Rosazza Gianin \cite{FRG02}, and
Heath and Ku \cite{HK04}, generalizing
Artzner et al. \cite{ADEH99}; see also the early Wald \cite{W50}, Huber \cite{H81}, Deprez and Gerber \cite{DG85},
Ben-Tal and Teboulle \cite{BT86,BT87}, and the more recent Carr, Geman and Madan \cite{CGM01},
Ruszczy\'nski and Shapiro \cite{RS06a}
and Ben-Tal and Teboulle \cite{BT07}.
F\"ollmer and Schied \cite{FS11,FS13} and Laeven and Stadje \cite{LS13,LS14} provide precise connections between the two strands of the literature.
We explore the combination of optimal risk sharing and an expanding pool of risks
in the presence of uncertainty about the true probabilistic model.

More specifically, we start in this paper by considering classical expected utility,
so that the certainty equivalent $U$ of a risk $X$ is given by
\begin{equation}\label{eq:introce}U(X)=u^{-1}(\E{u(X)}),\end{equation}
with $u$ a utility function and $\mathrm{E}[\cdot]$ the expectation under an objectively or subjectively given probabilistic model.
We analyze in this setting the precise limiting behavior and convergence rates of the risk premium associated with the average risk
$S_{n}/n$, where $S_n=\sum_{i=1}^n X_i$ for i.i.d.\ risks $X_{i}$, $n\in \Nat$, given by
\begin{equation}\label{eq:introrp}\pi(v,S_{n}/n)=\E{v+S_{n}/n}-U(v+S_{n}/n),\end{equation}
corresponding to proportional (equal, $1/n$) risk sharing of the aggregate risk
among $n$ cooperative individuals with identical utility function $u$
and initial wealth $v$,
which we prove to be Pareto optimal in this setting under mild conditions.

Next, we explicitly take uncertainty about the probabilistic model into account
and adopt a robust approach.
This setting turns out to be intriguingly more delicate.
It is best thought of as featuring probabilistic models that are the Kolmogorov extensions
of a family of product probability measures.
We first consider certainty equivalents that are
``robustified'' over a class of such probabilistic models $\mathcal{P}$:
\begin{equation}\label{eq:introrobustce}
 \Ucal_\Pcal(X)=\inf_{Q\in \Pcal} U_Q(X)+\alpha(Q),
\end{equation}
with $U_Q(X)=u^{-1}(\EQ{u(X)})$
and where $\alpha:\mathcal{P}\rightarrow\mathbb{R}\cup\{\infty\}$ is a penalty function
that measures the plausibility of the probabilistic model $Q\in\mathcal{P}$.
We prove that the proportional risk sharing rule remains Pareto optimal in this setting.
Furthermore, we prove that in an expanding pool of risks
the robustified certainty equivalent of the average risk converges to the robustified expectation,
and we provide explicit bounds on the corresponding convergence rates.
We find in particular that the convergence rates are dictated
by the individuals' coefficient of absolute risk aversion and the robustified first two moments,
expectation and variance.

Finally, we naturally extend the risk premium of Pratt \cite{P64}
to our setting with risk \textit{and} ambiguity, by considering
\begin{equation*}
\pi(v,X)= \Wcal(v+X) - \Wcal_\Pcal(v+X)\quad\mathrm{and}\quad \pi(v,X)=\Ucal(v+X) - \Vcal_\Pcal(v+X),
\end{equation*}
in the case of homothetic and variational preferences, respectively,
with
\begin{equation*}
\Wcal(X)=\inf_{Q\in \Pcal} \EQ{X}\beta(Q)\quad\mathrm{and}\quad\Wcal_\Pcal(X)=u^{-1}\left(\inf_{Q\in \Pcal} \EQ{u(X)}\beta(Q)\right),
\end{equation*}
where $\beta:\mathcal{P}\rightarrow[1,\infty]$ is a penalty function, 
and
\begin{equation*}
\Ucal(X)=\inf_{Q\in \Pcal} \EQ{X}+\alpha(Q)\quad\mathrm{and}\quad\Vcal_\Pcal(X)=u^{-1}\left(\inf_{Q\in \Pcal} \EQ{u(X)}+\alpha(Q)\right).
\end{equation*}
The robustified certainty equivalents and risk premia compound risk and ambiguity aversion (Ghirardato and Marinacci \cite{GM02}).
We prove that under Pareto optimal risk sharing in an expanding pool of risks the robust risk premium converges to zero in the homothetic case,
but, for non-trivial $\alpha$, will not vanish in the limit in the variational case, in which case it converges to
$\Ucal(v+X_1)-\Vcal(v+X_1)$,
with
\begin{equation*}
\Vcal(X)=u^{-1}\left(\inf_{Q\in \Pcal} u(\EQ{X})+\alpha(Q)\right),
\end{equation*}
and we analyze the corresponding convergence rates.

Our convergence results may be compared to the convergence results obtained by F\"ollmer and Knispel \cite{FK11a,FK11b}.
These authors analyze the limiting behavior of the risk capital 
per financial position,
when computing capital requirements for large and expanding portfolios of i.i.d.\ financial positions, 
in the \textit{absence} of optimal risk sharing
(and in the more restrictive setting of convex measures of risk rather than
the general setting provided by homothetic and variational preferences, as is considered here).
This seemingly related problem
requires much different techniques
and leads to completely different results.
For example, without optimal risk sharing, the certainty equivalent per position
in an expanding portfolio of $n$ i.i.d.\ risks under expected utility with exponential utility
(which yields a prototypical example of a convex measure of risk, up to a sign change)
is constant in $n$.
By contrast, with Pareto optimal risk sharing,
the certainty equivalent of the average risk in an expanding pool of $n$ i.i.d.\ risks under
expected utility with exponential utility
(and the same coefficient of absolute risk aversion for all $n$ individuals)
converges to the plain expectation, as $n$ tends to infinity.

The remainder of this paper is organized as follows.
In Section \ref{sec:cert:equiv} we analyze the asymptotic behavior of the certainty equivalents
and risk premia with optimal risk sharing in an expanding pool under the expected utility model.
In Section \ref{sec:cert:rob} we study optimal risk sharing under ambiguity
and investigate the asymptotic behavior of the robustified certainty equivalents.
In Section \ref{sec:rp:rob} we consider robust risk premia
and analyze their limits and convergence rates.
Conclusions are in Section \ref{sec:con}.

\setcounter{equation}{0}

\section{Optimal Risk Sharing: Certainty Equivalents and Risk Premia}\label{sec:cert:equiv}


Let $u:\R\to [-\infty, \infty)$ be a utility function.\footnote{We allow $u$ to take the value $-\infty$ in order to incorporate utility functions with bounded domains such as power utilities or logarithmic utilities, etc.}
Unless explicitly stated otherwise, we assume all utility functions considered in this paper to be strictly increasing on their domain $\dom u:=\{x\in \R\mid u(x)>-\infty\}$, concave and twice continuously differentiable on the interior of their domain $\operatorname{int} \dom u$.
We denote by $u^{-1}$ the inverse of $u$.
The inverse is well-defined on the image $\operatorname{Im} u \cap \R$ since $u$ is strictly increasing on $\dom u$,
and we extend $u^{-1}$ to $\operatorname{Im} u$ by setting $u^{-1}(-\infty):=-\infty$.
Moreover, we assume that $u^{-1}$ is continuously differentiable on $\operatorname{int}\dom u^{-1}$.

Let $(\Omega,{\cal F},P)$ denote some fixed probability space and $\E{\cdot}$ the expectation with respect to $P$.
Furthermore, let $X\in L^1:=L^1(\Omega,{\cal F},P)$\footnote{Throughout the paper, for the sake of brevity, we will stick to the convention of not differentiating between random variables and the $P$-almost sure equivalence classes they induce.}.
Consider an agent whose preferences are described by the expected utility criterion $\E{u(X)}$, with a subjective utility function $u$.
The \textit{certainty equivalent} corresponding to $u$ is given by
\begin{equation}\label{ce:eq2} U(X):=u^{-1}(\E{u(X)}),\quad  X\in L^1,\end{equation}
which takes values in $[-\infty,\infty)$ (by Jensen's inequality).
Suppose the agent has initial wealth $v\geq 0$ and considers taking an additional risk $X$.
The associated \textit{risk premium}, $\pi(v,X)$, is then obtained as
the solution to the equivalent utility equation
\begin{equation}\label{eq:risk:prem}u(v+\E{X}-\pi(v,X))=\E{u(v+X)}, \; \mbox{i.e.,}\quad \pi(v,X)=v+\E{X}-U(v+X).\end{equation}
The risk premium makes the agent indifferent between taking the risk on the one hand
and earning the expectation of the risk minus the risk premium with certainty on the other.

A special case of interest occurs if we consider exponential utility, $u(x)=1-\exp(-\gamma x)$, $\gamma>0$,
which exhibits constant absolute risk aversion (CARA), because $-u''(x)/u'(x)=\gamma$.
Then the certainty equivalent is given by
\begin{equation}\label{eq:entropicce}U(X)=-\frac{1}{\gamma}\log \E{\exp(-\gamma X)},\end{equation}
with $U(X)=\E{X}$ when $\gamma\downarrow 0$ and $U(X)=\essinf{X}$ when $\gamma\uparrow\infty$,
and which is non-increasing in $\gamma$.
It corresponds to minus the entropic measure of risk
(F\"ollmer and Schied \cite{FS11}), 
or minus the exponential premium for the loss $-X$ (Gerber \cite{G79}, Goovaerts, de Vylder and Haezendonck \cite{GdVH84} and Goovaerts et al. \cite{GKLT04}).
It is particularly popular in decision theory (see e.g., Gollier \cite{G01})
and financial mathematics (see e.g., Rouge and El Karoui \cite{ReK00}, Mania and Schweizer \cite{MS05} and the references therein). 
The corresponding risk premium is given by
\begin{equation*}\pi(v,X)=\E{X}+\frac{1}{\gamma}\log \E{\exp(-\gamma X)},\end{equation*}
which is independent of $v$.

Optimal risk sharing under the expected utility model was studied e.g., by
Borch \cite{B62}, Wilson \cite{W68}, DuMouchel \cite{D68}, Gerber \cite{G78,G79}, B\"uhlmann and Jewell \cite{BJ79}
and Gerber and Pafumi \cite{GP98}.
We consider a market or pool of $n$ expected utility maximizers with identical utility function $u$
and with aggregate random endowment $W$.
We are interested in the problem of finding the ``most efficient'' subdivision of $W$
among the $n$ agents.
We let ${\mathbb A}(W):=\{(Y_1,\ldots, Y_n)\in L^1\mid \sum_iY_i=W\}$ be the set of all possible (full) allocations of $W$.
The following lemma and proposition prove Pareto optimality and uniqueness of the proportional (equal, $1/n$) risk sharing rule
under mild conditions:
\begin{lemma}\label{lem:exput1}
Suppose that all $n$ agents apply the same expected utility criterion $\E{u(X)}$, $X\in L^1$, and that the aggregate random endowment $W$ satisfies $W/n\in \operatorname{int} \dom u$ $P$-a.s.\ and $\E{u(W/n)}\in \R$.
Then, the allocation which assigns the share $W/n$ of $W$ to each agent is Pareto optimal.
Indeed we have that
\begin{equation}\label{eq:unique:exput}
n\E{u\left(\frac{W}{n}\right)} = \max_{(Y_1,\ldots, Y_n)\in {\mathbb A}(W)}\sum_{i=1}^n \E{u(Y_i)}.
\end{equation}
Furthermore, if $u$ is strictly concave on $\dom u$, then the allocation $(W/n,\ldots, W/n)$ is the only solution to \eqref{eq:unique:exput}.
\end{lemma}

\begin{proof}
Let $(Y_1,\ldots,Y_n)$ be an allocation of $W$ such that
$\sum_{i=1}^n \E{u(Y_i)}\in \R$ (clearly, \eqref{eq:unique:exput} is trivially satisfied if $\sum_{i=1}^n \E{u(Y_i)}=-\infty$).
Then we must have $Y_i\in \dom u$ $P$-a.s.\ for all $i=1,\ldots, n$.
Now apply Jensen's inequality to obtain:
\begin{equation*}\frac{1}{n}\sum_{i=1}^n \E{u(Y_i)}\leq \E{u\left(\frac1n\sum_{i=1}^nY_i\right)}=\E{u\left(\frac{W}{n}\right)}.\end{equation*}
The final assertion about uniqueness is a direct consequence of the strict concavity of $u$.
\end{proof}

\begin{proposition}\label{prop:cert1}
Suppose that all $n$ agents apply the same certainty equivalent criterion $U$ as in \eqref{ce:eq2}.
Also suppose that the aggregate random endowment $W$ satisfies $W/n\in \operatorname{int} \dom u$ $P$-a.s.\ and $\E{u(W/n)}\in \R$. Then, the allocation which assigns the share $W/n$ of $W$ to each agent is Pareto optimal
from the perspective of the certainty equivalents.
Furthermore, if $n\geq 3$, $\dom u= \R$, and $\E{u(W)}\in \R$, then
\begin{equation}\label{eq:convolution} nU\left(\frac{W}{n}\right)=\max_{(Y_1,\ldots, Y_n)\in {\mathbb A}(W)}\sum_{i=1}^n U(Y_i),\end{equation}
and any Pareto optimal allocation of $W$ must maximize the right hand side of \eqref{eq:convolution}.
In that case, if $U$ is strictly concave,
\begin{equation*}\mbox{ i.e.,\ if $P(X\neq Y)>0$ implies $U(\lambda X+(1-\lambda)Y)>\lambda U(X)+(1-\lambda)U(Y)$ for all $\lambda\in (0,1)$,}
\end{equation*}then $W/n$ is the unique Pareto optimal allocation of $W$.
If $U$ is the certainty equivalent given by (\ref{eq:entropicce}), then $W/n$ is the unique Pareto optimal allocation of $W$ up to a reallocation of cash,
i.e., all Pareto optimal allocations are elements of $\{(W/n+m_1,\ldots, W/n+m_n)\mid \sum_{i=1}^n m_i=0\}$.
\end{proposition}

\begin{proof}
The fact that $W/n$ is Pareto optimal from the perspective of the certainty equivalents
follows from Lemma \ref{lem:exput1} and the fact that $u^{-1}$ is strictly increasing on $\dom u^{-1}=\operatorname{Im} u\cap \R$.
Next, for any Pareto optimal allocation $(X_1,\ldots, X_n)$ there are weights $\lambda_1, \ldots, \lambda_n\geq 0$, not all equal to $0$, such that \begin{equation*}\sum_{i=1}^n\lambda_iU(X_i)=\sup_{(Y_1,\ldots, Y_n)\in {\mathbb A}(W)}\sum_{i=1}^n\lambda_iU(Y_i); \end{equation*}
see e.g., Gerber \cite{G79}.
Suppose that $n\geq 3$, $\dom u= \R$, and $\E{u(W)}\in \R$.
Then, in particular  $U(m)=m$ for all $m\in \R$.
Suppose that  $\lambda_1>\lambda_2$, then
\begin{eqnarray*}
\sum_{i=1}^n\lambda_iU(X_i)& = & \sup_{(Y_1,\ldots, Y_n)\in {\mathbb A}(W)}\sum_{i=1}^n\lambda_iU(Y_i) \\ &\geq & \sup_{m>0} (\lambda_1-\lambda_2)m +\lambda_3 U(3W/n) + (n-3)U(W/n),
\end{eqnarray*} where $U(3W/n)>-\infty$ follows from concavity of $u$ and $\E{u(W)}\in \R$.
Letting $m\to \infty$ yields a contradiction. Hence, $\lambda_1=\lambda_2$, and similarly it follows that indeed $\lambda_1= \lambda_2=\ldots = \lambda_n$. Consequently we may assume that $\lambda_i=1$ for all $i$. This proves \eqref{eq:convolution}. Now suppose that $U$ is in addition strictly concave, and that there is another Pareto optimal allocation $(X_1,\ldots, X_n)$ of $W$ different from $(W/n,\ldots, W/n)$. Then there is at least one $j\in \{1,\ldots, n\}$ such that $P(X_i\neq W/n)>0$.
By strict concavity of $U$, we have that
$$\sum_{i=1}^nU\left(\frac12 (X_i+W/n)\right)>\frac12 \sum_{i=1}^n \left(U(X_i)+U(W/n)\right) = n U(W/n),$$ which contradicts \eqref{eq:convolution}.
If $U$ is the certainty equivalent given by (\ref{eq:entropicce}), then $U$ is strictly concave up to constants in the sense that if $X-Y\not \in \R$, then $U(\lambda X+(1-\lambda)Y)>\lambda U(X)+(1-\lambda)U(Y)$ for all $\lambda\in (0,1)$.
This proves the last assertion.
\end{proof}

In the remainder of this section, we let $X_i\in L^1$, $i=1,2,\ldots$,
be an i.i.d.\ sequence of risks and let $S_n$ given by $S_n=\sum_{i=1}^n X_i$, $n\in \Nat$, be the aggregate risk in a pool.
Then, $S_n/n\to \E{X_1}$ $P$-a.s.\ as $n\to \infty$ by the strong law of large numbers.
With this benchmark result in mind,
and using Lemma \ref{lem:exput1} and Proposition \ref{prop:cert1},
we analyze the behavior 
of the certainty equivalents and risk premia,
induced by Pareto optimal risk sharing of the aggregate risk among $n$ cooperative agents with identical expected utility preferences,
as the pool expands.

\begin{lemma}\label{lem:upperbound}
Consider a certainty equivalent $U$ as in $\eqref{ce:eq2}$. Then $U(S_n/n)\leq \E{X_1}$ and $U(S_n/n)$ is increasing in $n$. Thus $\pi(v, \frac{S_n}{n})$ is decreasing in $n$.
\end{lemma}

\begin{proof}
It follows from Jensen's inequality that always $U(S_n/n)\leq \E{X_1}$. Next we show that  $U(S_n/n)$ is increasing in $n$. Indeed, we may rewrite \begin{equation*}S_{n+1}=\frac{1}{n}\sum_{i=1}^{n+1}S_n^i,\quad \mbox{where} \quad S_n^i:=\sum_{j=1,\, j\neq i}^{n+1} X_j,\end{equation*}
and thus by concavity of $u$,
\begin{equation}\label{eq:conc:order}
\E{u\left(\frac{S_{n+1}}{n+1}\right)} = \E{u\left(\frac{1}{n+1}\sum_{i=1}^{n+1}\frac{S^i_{n}}{n}\right)}
\geq \frac{1}{n+1}\sum_{i=1}^{n+1}\E{u\left(\frac{S^i_{n}}{n}\right)}=\E{u\left(\frac{S_{n}}{n}\right)},
\end{equation}
because $S_n^i$ and $S_n$ are identically distributed under $P$.
Hence, $U(S_n/n)\leq U(S_{n+1}/(n+1))$ since $u^{-1}$ is increasing.
The final statement follows from applying the first statement of the lemma to the i.i.d.\ sequence $v+X_i$, and recalling the definition of the risk premium in \eqref{eq:risk:prem}.
\end{proof}

Lemma~\ref{lem:upperbound} shows that optimally pooling and relocating the aggregate risk reduces the risk premium.
Moreover, the following result shows that typically $U(S_n/n)\to \E{X_1}$ as $n\to \infty$,
i.e., the risk premium $\pi(v,\frac{S_n}{n})$ vanishes (converges to $0$) in the limit:

\begin{proposition}\label{ce:thm2} Suppose that $X_1\in L^2$, $X_1\in \op{int}\dom u$ $P$-a.s., $\E{u(X_1)}\in \R$, and that $u'(S_n/n)$ is bounded from above by some square integrable random variable independent of $n$.
Then,
\begin{equation}\label{ce:prop:eq2}
\limsup_{n\to \infty}\,\sqrt{n}\pi\left(v,\frac{S_n}{n}\right)=\limsup_{n\to \infty}\, \sqrt{n}\left(v+\E{X_1}-U\left(v+\frac{S_n}{n}\right)\right)\leq \sigma_P(X_1),
\end{equation}
where $\sigma_P(X_1):=\sqrt{\E{(X_1-\E{X_1})^2}}$ denotes the standard deviation of $X_1$.
In particular, \eqref{ce:prop:eq2} holds already in case $u$ is only once continuously differentiable (instead of twice).
\end{proposition}

However, under some additional assumptions we can say even more about the rate of convergence:

\begin{theorem}\label{ce:thm1}
Suppose that $X_1\in L^4$, $X_1\in \op{int}\dom u$ $P$-a.s., $\E{u(X_1)}\in \R$, and that $u''$ can be controlled from below by some $\widetilde Y\in L^2$ on the range of $S_n/n$ in the following way: \begin{equation}\label{eq:control}\widetilde Y\leq \essinf\left\{u''(v+Y)\mid \mbox{$Y$ is a random variable s.t.\ } Y(\omega)\in \left[\E{X_1}, \frac{S_n(\omega)}{n}\right] \;P\mbox{-a.s.}\right\},\end{equation}
with the convention that for $a<b$ we set $[b,a]:=[a,b]$.
Then,
\begin{equation}\label{ce:prop2:eq2}
\lim_{n\to \infty}\, n\pi\left(v,\frac{S_n}{n}\right) = \lim_{n\to \infty}\, n\left(v+\E{X_1}-U\left(v+\frac{S_n}{n}\right)\right)= \frac12 R(v+\E{X_1})\sigma_P^2(X_1),
\end{equation} where $R(x):=-u''(x)/u'(x)$ is the Arrow-Pratt coefficient of absolute risk aversion.
\end{theorem}

Note that the bound on the right hand side of \eqref{ce:prop:eq2} does not depend on the utility function $u$,
contrary to the far right hand side of \eqref{ce:prop2:eq2}.
Also, note that \eqref{eq:control} is always satisfied if, for instance, $X_1$ with $\essinf X_1\in \operatorname{int} \dom u$ is bounded,
since $u''$ is continuous and thus bounded on the compact set $\{v\}+[\essinf X_1,\esssup X_1]$.

\begin{proof} [of Proposition~\ref{ce:thm2} and Theorem~\ref{ce:thm1}] Without loss of generality we may assume that $v=0$, because if $X_1$ satisfies the requirements of Proposition~\ref{ce:thm2} or Theorem~\ref{ce:thm1}, so does $X_1+v$, and $\sigma_P(X_1+v)=\sigma_P(X_1)$.
We compute a Taylor expansion of $u$ around $\E{X_1}$, either to the first or to the second order, which yields
\begin{equation}\label{eq1:alt}u\left(\frac{S_n}{n}\right)=u(\E{X_1})+u'(Y_n)\left(\frac{S_n}{n}-\E{X_1}\right),\end{equation}
and
\begin{equation}\label{eq1}u\left(\frac{S_n}{n}\right)=u(\E{X_1})+u'(\E{X_1})\left(\frac{S_n}{n}-\E{X_1}\right)+\frac12 u''(Z_n)\left(\frac{S_n}{n}-\E{X_1}\right)^2,\end{equation}
where $Y_n$ and $Z_n$ are random variables taking values between $\E{X_1}$ and $\frac{S_n}{n}$.
(Note that $S_n/n\in \op{int}\dom u$ as $X_1\in \op{int}\dom u$.)
Taking expectations in \eqref{eq1:alt} and \eqref{eq1} we arrive at
\begin{equation*}\label{eq2:alt}\E{u\left(\frac{S_n}{n}\right)}=u(\E{X_1})+\E{u'(Y_n)\left(\frac{S_n}{n}-\E{X_1}\right)},\end{equation*}
and
\begin{equation*}\label{eq2}\E{u\left(\frac{S_n}{n}\right)}=u(\E{X_1})+0+\frac12 \E{u''(Z_n)\left(\frac{S_n}{n}-\E{X_1}\right)^2}.\end{equation*}
Invoking a Taylor expansion of $u^{-1}$ around the point $u(\E{X_1})$ verifies that
\begin{eqnarray}U\left(\frac{S_n}{n}\right)&=&u^{-1}\left(\E{u\left(\frac{S_n}{n}\right)}\right)\nonumber \\
&=&u^{-1}\left(u(\E{X_1})+ \E{u'(Y_n)\left(\frac{S_n}{n}-\E{X_1}\right)}\right)\nonumber \\
&=&u^{-1}\circ u(\E{X_1})+(u^{-1})'(y_n) \E{u'(Y_n)\left(\frac{S_n}{n}-\E{X_1}\right)}\nonumber \\
&=&\E{X_1}+(u^{-1})'(y_n)\E{u'(Y_n)\left(\frac{S_n}{n}-\E{X_1}\right)},\label{eq3:alt}\end{eqnarray}
for some real number $y_n$ between $u(\E{X_1})$ and $u(\E{X_1})+ \E{u'(Y_n)(\frac{S_n}{n}-\E{X_1})}$,
and
\begin{eqnarray}U\left(\frac{S_n}{n}\right)&=&u^{-1}\left(\E{u\left(\frac{S_n}{n}\right)}\right)\nonumber \\
&=&u^{-1}\left(u(\E{X_1})+\frac12 \E{u''(Z_n)\left(\frac{S_n}{n}-\E{X_1}\right)^2}\right)\nonumber \\
&=&u^{-1}\circ u(\E{X_1})+(u^{-1})'(z_n)\frac12 \E{u''(Z_n)\left(\frac{S_n}{n}-\E{X_1}\right)^2}\nonumber  \\
&=&\E{X_1}+(u^{-1})'(z_n)\frac{u''(\E{X_1})\sigma^2_P(X_1)}{2n} \nonumber\\
&& \quad+ (u^{-1})'(z_n)\frac12 \E{(u''(Z_n)-u''(\E{X_1}))\left(\frac{S_n}{n}-\E{X_1}\right)^2},  \label{eq3}\end{eqnarray}
where $z_n\in [u(\E{X_1})+\frac12 \E{u''(Z_n)(\frac{S_n}{n}-\E{X_1})^2}, u(\E{X_1})]$.
By an application of H\"older's inequality to the error term in \eqref{eq3:alt} we obtain
\begin{equation}\label{eq:est2}
|(u^{-1})'(y_n)|\E{|u'(Y_n)\left(\frac{S_n}{n}-\E{X_1}\right)|} \leq |(u^{-1})'(y_n)|\sqrt{\E{u'(Y_n)^2}}\frac{\sigma_P(X_1)}{\sqrt{n}}.
\end{equation}
The dominated convergence theorem implies that $\sqrt{\E{u'(Y_n)^2}}\to u'(\E{X_1})$, because $Y_n\to \E{X_1}$ $P$-a.s.,\ and $0\leq u'(Y_n)\leq u'(\E{X_1})\vee u'(S_n/n)$ where $u'(S_n/n)$ is bounded by some square integrable random variable independent of $n$ by assumption. Also $y_n\to u(\E{X_1})$ and hence $$(u^{-1})'(y_n)=\frac{1}{u'(u^{-1}(y_n))}\to \frac{1}{u'(\E{X_1})},$$ which proves \eqref{ce:prop:eq2}.
If $X_1$ has finite fourth moment, then applying H\"older's inequality to the error term in \eqref{eq3} yields
\begin{align}
&  \E{|u''(Z_n)-u''(\E{X_1})|\left(\frac{S_n}{n}-E[X_1]\right)^2}\nonumber \\ \leq \quad & \sqrt{\E{\left(u''(Z_n)-u''(\E{X_1})\right)^2} }\frac{ \sqrt{nM_P(X_1)+3n(n-1)\sigma^2_P(X_1)^2}}{n^2}\nonumber \\ \leq\quad & \frac{2}{n} \sqrt{\E{\left(u''(Z_n)-u''(\E{X_1})\right)^2} } \sqrt{M_P(X_1)+\sigma^2_P(X_1)^2},\label{eq:est1}
\end{align} where $M_P(X_1):=\E{(X_1-\E{X_1})^4}$.
We have $\E{\left(u''(Z_n)-u''(\E{X_1})\right)^2} \to 0$ by dominated convergence since $|u''(Z_n)-u''(\E{X_1})|\leq |\widetilde Y| + |u''(\E{X_1})|$ for some square integrable $\widetilde Y$ by \eqref{eq:control}.
Also $z_n\to u(\E{X_1})$ for $n\to \infty$, and using $(u^{-1})'(z_n)=1/u'(u^{-1}(z_n))$, we conclude from \eqref{eq3} that $$n\left(U\left(\frac{S_n}{n}\right)-\E{X_1}\right)\to \frac12 \frac{u''(\E{X_1})}{u'(\E{X_1})}\sigma^2_P(X_1).$$
\end{proof}

\begin{remark}\label{rem:Pratt}(Relation to Pratt~\cite{P64}.)
In his seminal paper, Pratt \cite{P64} shows that the risk premium satisfies \begin{equation}\label{eq:pratt}\pi(v,X)=\frac{1}{2}R(v+\E{X})\sigma^2_P(X) + o(\sigma^2_P(X)).\end{equation}
However, note that \eqref{ce:prop2:eq2} does not follow from \eqref{eq:pratt} since \eqref{eq:pratt} implies $$ \pi\left(v,\frac{S_n}{n}\right)=\frac{1}{2n}R(v+\E{X_1})\sigma^2_P(X_1)+ o\left(\frac1n\sigma^2_P(X_1)\right),$$
and the estimate $o(\frac1n\sigma^2_P(X_1))$ is too rough.
\end{remark}

\begin{examples}\label{ex:cert:equiv}
Consider the exponential utility $u(x)=1-e^{-\gamma x}$ for some $\gamma>0$.
Then \eqref{ce:prop2:eq2} becomes
\begin{equation*}\label{eq:exexponential}\lim_{n\to \infty}\, n\, \pi\left(v, \frac{S_n}{n}\right)= \frac{1}{2}\gamma \sigma_P^2(X_1).\end{equation*}
For power utility $u(x)=\frac{x^{1-\chi}-1}{1-\chi}$, $x> 0$, where $\chi>0$, $\chi\neq 1$, (and $u(x)=-\infty$, $x\leq 0$), \eqref{ce:prop2:eq2} becomes
\begin{equation*}\label{eq:expower}\lim_{n\to \infty}\, n\, \pi\left(v, \frac{S_n}{n}\right)= \frac{\chi}{2(v+\E{X_1})} \sigma_P^2(X_1).\end{equation*}
In case of the logarithmic utility $u(x)=\log x$, $x> 0$, (and $u(x)=-\infty$, $x\leq 0$), \eqref{ce:prop2:eq2} becomes
\begin{equation*}\label{eq:exlog}\lim_{n\to \infty}\, n\, \pi\left(v,\frac{S_n}{n}\right)= \frac{1}{2(v+\E{X_1})} \sigma_P^2(X_1).\end{equation*}
\end{examples}

\begin{remark}\label{rem:speed}
Suppose that the i.i.d.\ assumption on the sequence of random variables (imposed just above Lemma \ref{lem:upperbound}) is not satisfied.
For instance, let $X$ be a standard normal random variable and consider the sequence $X_i=(-1)^iX$, $i\in \Nat$.
Then all $X_i$ are identically distributed but apparently not independent.
Clearly, $S_n=\sum_{i=1}^nX_i=-X$ whenever $n$ is odd and $S_n=0$ otherwise.
Hence, $S_n/n\to 0=\E{X_1}$ for $n\to \infty$.
Take $u(x)=1-e^{-\gamma x}$.
Then for odd $n$ we compute $nU(S_n/n)=-\gamma/(2n)$ whereas for even $n$ we have $nU(S_n/n)=0$.
Thus, the left hand side of \eqref{ce:prop2:eq2} for this case equals $\lim_{n\to \infty}-nU(S_n/n)=0$.
However, the right hand side of \eqref{ce:prop2:eq2} would be $\gamma/2$; see Examples~\ref{ex:cert:equiv}.
Hence, requiring independence of the sequence $X_i$ is crucial in Theorem~\ref{ce:thm1}.
(This counterexample also generalizes to the robust case considered in the next section.)
\end{remark}

\begin{example}\label{ex:entropic}(Optimal risk sharing for entropic measures of risk.)
Consider optimal risk sharing among $n$ agents that apply the same certainty equivalent criterion $U$ with exponential utility
and aggregate risk $S_n$.
Then,
\begin{equation*}\rho_{P,\gamma }(X):=-U(X)=\frac{1}{\gamma} \log\E{e^{-\gamma X}},\end{equation*}
is the so-called entropic risk measure.
According to Proposition \ref{prop:cert1}, the optimal relocation, $Y_{i}^*$,
is given by $Y_{i}^* = S_n / n, \; i=1,\ldots,n$.
It follows that, after optimal exchange and relocation of risks,
the cooperating pool of agents,
all using the same entropic measure of risk,
can also be seen to use the entropic measure of risk at the aggregate level, with parameter $\gamma_n:=\gamma/n$, in the sense that
\begin{align*}n\rho_{P,\gamma}(Y_{i}^*)&= n(1/\gamma) \log \E{\exp(-\gamma Y_{i}^*)}=n(1/\gamma) \log \E{\exp(-\gamma(1/n) S_n)}\\
&=(1/\gamma_n) \log \E{\exp(-\gamma_n S_n)} = \rho_{P,\gamma_n}(S_n).
\end{align*}
Hence, Proposition~\ref{ce:thm2} implies that under Pareto optimal risk sharing,
and thus different from F\"ollmer and Knispel \cite{FK11a,FK11b},
the pooling, exchange and relocation of risks has the effect that the gross entropic
risk measure per position decreases to minus the expectation, as follows:
\begin{equation*}
\lim_{n\rightarrow\infty} \frac{1}{n}\rho_{P,\gamma_n}(S_n)= \lim_{n\rightarrow\infty} \rho_{P,\gamma}(S_n/n)
= \E{-X_1}.
\end{equation*}
By contrast, in F\"ollmer and Knispel \cite{FK11a,FK11b}, who consider $\frac1n\rho_{P,\gamma}(S_n)$,
\begin{equation*}\frac1n\rho_{P,\gamma}(S_n)=\frac{1}{n}(1/\gamma) \log \E{\exp(-\gamma S_n)}= (1/\gamma) \log \E{\exp(-\gamma X_1)}=\rho_{P,\gamma}(X_1),
\end{equation*}
and hence \begin{equation*}\lim_{n\to \infty}\frac1n\rho_{P,\gamma}(S_n)= \rho_{P,\gamma}(X_1).\end{equation*}
\end{example}



\setcounter{equation}{0}

\section{Optimal Risk Sharing under Ambiguity: Robust Certainty Equivalents}\label{sec:cert:rob}



So far we have fixed a reference probabilistic model $P$ which is assumed to be (objectively or subjectively) known.
Let us now consider the 
situation of model uncertainty in which $P$ is replaced by a
class $\mathcal{P}$ of probabilistic models on the measurable space $(\Omega,\mathcal{F})$.
Such a situation occurs naturally, as follows.
The utility function $u$ is determined by preferences on the constants $\R$ only, and we thus assume that it is (subjectively) given.
The assumption that the sequence of risks $(X_i)_{i\in \Nat}$ is i.i.d.\ is typical and commonly adopted, of course depending on the collection of data,
and corresponds to the setting of the previous section.
However, the choice of the ``right'' probabilistic model/distribution may be a very delicate problem.
Hence, one may consider all probabilistic models such that the i.i.d.\ assumption on $(X_i)_{i\in \Nat}$ is satisfied,
maybe reduce this class further due to some additional probabilistic information,
and yet end up with a class $\Pcal$ of probabilistic models,
considered as possible generators of the observed $(X_i)_{i\in \Nat}$,
that contains more than only a single element.

In probabilistic model terms one may start with any class of probability measures $\widetilde \Pcal$
on some measurable space $(\Sigma, {\mathcal A})$ and corresponding distributions of the random variable $X$ on $(\Sigma, {\mathcal A})$ that one wants to consider.
Now one may think of $\Pcal$ as the collection of probability measures on the product space $(\Omega, {\cal F}):=(\Sigma^\Nat,{\cal A}^{\otimes \Nat})$ such that each $P\in \Pcal$ is the Kolmogorov extension for some $Q\in \widetilde \Pcal$ of the family of product probability measures $\{Q^{\otimes n}\mid n\in \Nat\}$, i.e.\ $P_{|{\cal A}^{\otimes n}}=Q^{\otimes n}$; see, for instance, Dudley \cite{D02}~Section~8.2.
Then, the sequence $X_i:\Omega^\Nat\to \R$ given by $X_i(\omega_1,\omega_2,\ldots):=X(\omega_i)$ is i.i.d.\ under each $P\in \Pcal$.


Typical examples of $\mathcal{P}$
are for instance:
\begin{examples}
\begin{itemize}
\item[(i)] A discrete set of probability measures.
    In this case, the agent considers finitely many probability measures $P_{j}\in\mathcal{P}$, $j=1,\ldots,J$, $J\in\mathbb{N}$.
    Such discrete sets are particularly popular in e.g.\ robust optimization and operations research.
\item[(ii)] Parametric families $\mathcal{P}=\{P_{\theta}\mid\theta\in\Theta\}$.
    In this case, the agent restricts attention to i.i.d.\ sequences $X_{i}$ with probability distribution that belongs to a given parametric family of probability distributions (e.g.\ a Beta distribution),
    which is indexed by a parameter vector $\theta$ (e.g.,\ the mean or the shape parameter(s)) in a set of possible parameters $\Theta$ of the full parameter space.
\end{itemize}
\end{examples}

Henceforth, we let $\Pcal$ be a non-empty set of probability measures on $(\Omega, {\cal F})$ (not necessarily dominated). In this section we restrict ourselves to the model space
\begin{equation*}L^\infty_\Pcal:= \{X:\Omega\to \R\mid \mbox{$X$ is $\cal F$-measurable}, \exists m>0:\; \forall P\in \Pcal:\, P(|X|\leq m)=1\}_{/\sim},\end{equation*}
where $X\sim Y$ if and only if $P(X=Y)=1$ for all $P\in \Pcal$.
$(L^\infty_\Pcal, \|\cdot\|_{\Pcal,\infty})$ where
\begin{equation*}\|X\|_{\Pcal,\infty}:=\inf\{m\in \R\mid \forall P\in \Pcal: \, P(|X|\leq m)=1\},\end{equation*}
is a Banach space.
Note that in case $\Pcal$ is dominated, i.e.,\ if there exists a probability measure $P$ on $(\Omega,{\cal F})$ such that $Q\ll P$ for all $Q\in \Pcal$,\footnote{Here, $Q\ll P$ means that for all $A\in {\cal F}$ we have that $P(A)=0$ implies $Q(A)=0$.} then $L^\infty_\Pcal$ may be viewed as a subset of $L^\infty(\Omega,{\cal F},P)$ with $L^\infty_\Pcal= L^\infty(\Omega,{\cal F},P)$ if $P\approx \Pcal$.\footnote{Here, $P\approx \Pcal$ means that for all $A\in {\cal F}$ we have $P(A)=0$ if and only if ($Q(A)=0$ for all $Q\in \Pcal$).}
In the sequel, we say that a property holds $\Pcal$-a.s.\ if the set $A\in {\cal F}$ of $\omega\in \Omega$ with that certain property satisfies $P(A)
=1$ for all $P\in \Pcal$.
For every $Q\in \Pcal$, we let
\begin{equation*}U_Q(X):=u^{-1}(\EQ{u(X)}), \quad X\in L^\infty_\Pcal,\end{equation*}
be the corresponding certainty equivalent under $Q$ associated with the utility function $u$.
In the following, we consider \textit{robust certainty equivalents} of the form (cf. Eqn. (2) in Laeven and Stadje \cite{LS13})
\begin{equation}\label{eq:robust}
\Ucal_\Pcal(X):=\inf_{Q\in \Pcal} U_Q(X)+\alpha(Q), \quad X\in L^\infty_\Pcal,
\end{equation}
where $\alpha:\Pcal\to \R\cup\{\infty\}$ is a penalty function such that $\inf_{Q\in \Pcal}\alpha(Q)>-\infty$.
The latter assumption ensures that for any $X\in L^\infty_\Pcal$
such that $$\essinf_\Pcal X:=\sup\{m\in \R\mid \forall P \in \Pcal:\, P(X\geq m)=1\} \in \dom u,$$
we have $$\Ucal_\Pcal(X)\geq \essinf_\Pcal X + \inf_{Q\in \Pcal}\alpha(Q)>-\infty,$$ in which we used the obvious estimate $U_Q(X)\geq \essinf_\Pcal X$.
The penalty function represents the esteemed plausibility of the probabilistic model.
It is also referred to as an ambiguity index.
Robust certainty equivalents compound risk and ambiguity aversion in the sense of Ghirardato and Marinacci \cite{GM02}.
Any robust certainty equivalent $\Ucal_\Pcal$ has an associated \textit{robust expectation} (or \textit{robust monetary utility}) $\Ucal$ given by
\begin{equation}\label{eq:ass:monut}
 \Ucal(X):=\inf_{Q\in \Pcal} \EQ{X}+\alpha(Q), \quad X\in L^\infty_\Pcal.
\end{equation}
Note that Jensen's inequality implies $\Ucal_\Pcal(X)\leq \Ucal(X)$ for all $X\in L^\infty_\Pcal$.

As a special case of interest, we consider the
robust version of the certainty equivalent under exponential utility,
given by
\begin{displaymath}
\inf_{Q\in\mathcal{P}}-\tfrac{1}{\gamma}\log \EQ{e^{-\gamma X}}=-\sup_{Q\in\mathcal{P}}\tfrac{1}{\gamma}\log \EQ{e^{-\gamma X}}.
\end{displaymath}
It equals minus the
robust version, $\rho_{\mathcal{P},\gamma}$, of the (convex) entropic risk measure, $\rho_{Q,\gamma}$, defined by
\begin{equation}\label{eq:ecmor}
\rho_{\mathcal{P},\gamma}(X):=\sup_{Q\in\mathcal{P}} \rho_{Q,\gamma}(X)=\tfrac{1}{\gamma}\sup_{Q\in\mathcal{P}}\log \EQ{e^{-\gamma X}},
\end{equation}
which belongs to the class of entropy coherent measures of risk (Laeven and Stadje \cite{LS13}); see also F\"ollmer and Knispel \cite{FK11a}.
More generally, one may consider the class of entropy convex measures of risk (Laeven and Stadje \cite{LS13}):
the mapping $\rho_{\mathcal{P},\gamma,\alpha}:L^\infty_\Pcal\rightarrow\mathbb{R}$ is called a $\gamma$-entropy convex measure of risk if there exists a penalty function $\alpha: \mathcal{P}\rightarrow[0,\infty]$ with $\inf_{Q\in\mathcal{P}}\alpha(Q)=0$, such that
\begin{equation}\label{eq:ecxmor}\rho_{\mathcal{P},\gamma,\alpha}(X)=\sup_{Q\in\mathcal{P}}\{\rho_{Q,\gamma}-\alpha(Q)\},\quad \gamma>0.\end{equation}

We state the following lemma, which proves that the proportional risk sharing rule
remains Pareto optimal in this general setting with ambiguity:
\begin{lemma}\label{lem:robust2} Consider $n$ agents with the same robust certainty equivalent criterion \eqref{eq:robust}.
Let $W\in L^\infty_\Pcal$ be such that $W/n\in \operatorname{int} \dom u$ $\Pcal$-a.s.\ and $\essinf_\Pcal W/n\in \dom u$. Moreover, suppose that \begin{equation}\label{eq:min1}
\Ucal_\Pcal(W/n)=\min_{Q\in \Pcal} U_Q(W/n)+\alpha(Q).\end{equation}
Then the allocation which assigns the share $W/n$ of $W$ to each agent is Pareto optimal.
\end{lemma}

\begin{proof}
Let $(Y_1,\ldots, Y_n)$ be an allocation of $W$ such that $\Ucal_\Pcal(Y_i)\geq \Ucal_\Pcal(W/n)$ for all $i=1,\ldots,n$.
By assumption, there exists a $Q\in \Pcal$ such that $$\Ucal_\Pcal(W/n)=U_Q(W/n)+\alpha(Q).$$ Then $U_Q(Y_i)\geq U_Q(W/n)$ for all $i=1,\ldots,n$ which implies $U_Q(Y_i)= U_Q(W/n)$ for all $i=1,\ldots,n$ according to Proposition~\ref{prop:cert1}. Hence, also $$\Ucal_\Pcal(Y_i)\leq U_Q(Y_i)+\alpha(Q)= U_Q(W/n)+\alpha(Q)=\Ucal_\Pcal(W/n),$$ so indeed $\Ucal_\Pcal(Y_i)= \Ucal_\Pcal(W/n)$ for all $i=1,\ldots,n$.
\end{proof}

The next lemma stipulates a situation in which condition \eqref{eq:min1} is automatically satisfied.
\begin{lemma}\label{lem:robust1}
Suppose that $\Pcal$ is dominated by some probability measure $P$ on $(\Omega,{\cal F})$.
Moreover, suppose that $\Pcal$ is weakly compact in the sense that the set of densities $\left\{\frac{dQ}{dP}\mid Q\in \Pcal  \right\}$ is weakly compact, i.e., $\sigma(L^1(\Omega,{\cal F},P),L^\infty(\Omega,{\cal F},P))$-compact, and that $\alpha$ is weakly lower semi-continuous in the sense that the lower level sets of densities $$E_k:=\left\{\frac{dQ}{dP}\mid Q\in \Pcal,  \alpha(Q)\leq k\right\},\quad k\in \R,$$ of the penalty function $\alpha$ in \eqref{eq:robust} are closed in $\sigma(L^1(\Omega,{\cal F},P),L^\infty(\Omega,{\cal F},P))$. Then,
\begin{equation*}
\Ucal_\Pcal(X)=\min_{Q\in \Pcal} U_Q(X)+\alpha(Q) \quad \mbox{for all $X\in L^\infty(\Omega,{\cal F},P)$ such that $\essinf_{\{P\}} X\in \dom u$},
\end{equation*}
and
\begin{equation*}
 \Ucal(X)=\min_{Q\in \Pcal} \EQ{X}+\alpha(Q)\quad \mbox{for all $X\in L^\infty(\Omega,{\cal F},P)$.}
\end{equation*}
\end{lemma}

\begin{proof}
Let $X\in L^\infty(\Omega,{\cal F},P)$ with $\essinf_{\{P\}} X\in \dom u$.
Choose a sequence $(Q_n)_{n\in \Nat}\subset \Pcal$
such that $$\Ucal_\Pcal(X)=\lim_{n\to \infty}(U_{Q_n}(X)+\alpha(Q_n)).$$ As $\Pcal$ is weakly compact,
(by the Eberlein-Smulian theorem; see e.g., Dunford and Schwartz \cite{DS58}) there is a subsequence which for simplicity we also denote by $(Q_n)_{n\in \Nat}$ such that $\frac{dQ_n}{dP}$ converges weakly to $\frac{dQ}{dP}$ for a $Q\in \Pcal$. Hence, $$\EQ{u(X)}={\rm E}_P\left[\frac{dQ}{dP}u(X)\right]= \lim_{n\to \infty}{\rm E}_P\left[\frac{dQ_n}{dP}u(X)\right]= \lim_{n\to \infty}{\rm E}_{Q_n}\left[u(X)\right],$$ because $u(X)\in L^\infty(\Omega,{\cal F},P)$. By lower semicontinuity of $\alpha$
we have  $\alpha(Q)\leq \liminf_{n\to \infty} \alpha(Q_n)$.
Thus we conclude that $$U_{Q}(X)+\alpha(Q)\leq \liminf_{n\to\infty} U_{Q_n}(X)+\alpha(Q_n)=\Ucal_\Pcal(X),$$ and therefore $$\Ucal_\Pcal(X)=U_{Q}(X)+\alpha(Q).$$
The result for $\Ucal$ follows similarly.
\end{proof}

In the remainder of this section, we analyze the asymptotic behavior of the robust certainty equivalents
associated with the Pareto optimal risk sharing contract,
as the pool expands to include a growing multitude of risks.

\begin{proposition}\label{prop:robust:conv}
Let $(X_i)_{i\in \Nat}\subset L^\infty_\Pcal$ be i.i.d.\ under all $Q\in \Pcal$, and suppose that $\essinf_\Pcal X_1\in \op{int} \dom u$.
Also, let $S_n:=\sum_{i=1}^n X_i$.
Then $\Ucal_\Pcal(S_n/n)$ is increasing in $n$ with
\begin{equation*}\Ucal_\Pcal\left(\frac{S_n}{n}\right)\leq \Ucal(X_1),\end{equation*}
and
\begin{equation*}\lim_{n\to \infty}\Ucal_\Pcal\left(\frac{S_n}{n}\right)=\Ucal(X_1),\end{equation*}
where $\Ucal$ is the robust monetary utility associated to $\Ucal_\Pcal$.
Moreover, there is a constant $K$ depending on $X_1$ and $u$ such that
\begin{equation*}\limsup_{n\to \infty}\,\sqrt{n}\left(\Ucal(X_1)-\Ucal_\Pcal\left(\frac{S_n}{n}\right)\right)\leq K.\end{equation*}
\end{proposition}

\begin{proof} First of all, the facts that $\Ucal_\Pcal(S_n/n)$ is increasing in $n$ and $\Ucal_\Pcal(S_n/n)\leq \Ucal(X_1)$ follow from the facts that $U_Q(S_n/n)$ is increasing in $n$ and $U_Q(S_n/n)\leq \EQ{X_1}$ for all $Q\in\Pcal$; see Lemma~\ref{lem:upperbound}.

\smallskip\noindent
Let $Q_n\in \Pcal$ such that $U_{Q_n}(\frac{S_n}{n})+\alpha(Q_n)\leq\Ucal_\Pcal(\frac{S_n}{n})+\frac{1}{n^2}$, $n\in \Nat$. Then
\begin{eqnarray}
\Ucal(X_1)-\Ucal_\Pcal\left(\frac{S_n}{n}\right)&\leq& {\rm E}_{Q_n}[X_1] -U_{Q_n}\left(\frac{S_n}{n}\right)+\frac{1}{n^2} \nonumber \\
&\leq & \frac{1}{n^2} + \sup_{Q\in \Pcal} \left\{\EQ{X_1} -U_{Q}\left(\frac{S_n}{n}\right)\right\}. \label{eq:cont1}
\end{eqnarray}
Recalling \eqref{eq:est2}, we further estimate:
\begin{eqnarray*}
\sup_{Q\in \Pcal } \EQ{X_1} -U_{Q}\left(\frac{S_n}{n}\right)  &\leq & L\,\sup_{Q\in \Pcal } \frac{\sigma_Q(X_1)}{\sqrt{n}}\quad \leq \quad  L\frac{2\|X_1\|_{\Pcal,\infty}}{\sqrt{n}},
\end{eqnarray*}
where we used H\"older in the first inequality and $$L:=(u^{-1})'(u(\esssup_\Pcal X_1)+u'(\essinf_\Pcal X_1)2\|X_1\|_{\Pcal,\infty})u'(\essinf_\Pcal X_1)$$ is a constant.
(Here $\esssup_\Pcal$ is defined analogously to $\essinf_\Pcal$ above.)
In the rough estimate leading to $L$ we used the fact that the function  $$\op{int}\dom u\ni x\mapsto (u^{-1})'(x)=\frac{1}{u'(u^{-1}(x))}$$ is  positive and increasing and the known bounds for $Y_n$ and $y_n$ in \eqref{eq:est2}.
\end{proof}

Proposition~\ref{prop:robust:conv} shows that, also in this general robust framework that explicitly accounts for ambiguity,
optimally pooling and relocating risk reduces the difference between the (robust) expectation
and the (robust) certainty equivalent,
\begin{equation*}\Ucal(X_1)-\Ucal_\Pcal(S_n/n);\end{equation*}
it eventually vanishes in the limit as $n\rightarrow\infty$.

\begin{theorem} \label{prop:robust:est}
Let $(X_i)_{i\in \Nat}\subset L^\infty_\Pcal$ be i.i.d.\ under all $Q\in \Pcal$, and suppose that $\essinf_\Pcal X_1\in \op{int}\dom u$.
Also, let $S_n:=\sum_{i=1}^n X_i$.
Furthermore, suppose that $u$ is three times continuously differentiable on $\op{int} \dom u$ whereas $u^{-1}$ is two times continuously differentiable on  $\op{int} \dom u^{-1}$.
Then,
\begin{equation}\label{eq:robust:est1}\limsup_{n\to \infty}\, n\left(\Ucal(X_1)-\Ucal_\Pcal\left(\frac{S_n}{n}\right)\right)\leq \frac12 \sup_{Q\in \Pcal}R(\EQ{X_1})\sigma^2_Q(X_1).\end{equation}
Moreover, if $Q\in \Pcal$ satisfies $\Ucal(X_1)=\EQ{X_1}+\alpha(Q)$
(see Lemma~\ref{lem:robust1}),
then \begin{equation}\label{eq:robust:est2}\liminf_{n\to \infty}\, n\left(\Ucal(X_1)-\Ucal_\Pcal\left(\frac{S_n}{n}\right)\right)\geq \frac12 R(\EQ{X_1})\sigma^2_Q(X_1).\end{equation}
\end{theorem}

Note that $\sup_{Q\in \Pcal}R(\EQ{X_1})\sigma^2_Q(X_1)\leq L4\|X_1\|_{\Pcal,\infty}^2<\infty$  where $L>0$ is an upper bound of the continuous function $\op{int}\dom u\ni x\mapsto R(x)$ on the compact set $[\essinf_\Pcal X_1, \esssup_\Pcal X_1]$.

\begin{proof} First we prove \eqref{eq:robust:est1}. To this end recall the proof of Proposition~\ref{prop:robust:conv} and in particular \eqref{eq:cont1}. Also recall \eqref{eq3} and \eqref{eq:est1}. As $X_1$ is bounded we estimate the last term in \eqref{eq3} for $Q\in \Pcal$ using \eqref{eq:est1} in the following way:
\begin{align*}
\Delta(n,Q):= \quad & \left|(u^{-1})'(z^Q_n)\frac12 \EQ{(u''(Z^Q_n)-u''(\EQ{X_1}))\left(\frac{S_n}{n}-\EQ{X_1}\right)^2}\right|\\ \leq \quad & 2L\frac{(4\|X_1\|_{\Pcal,\infty}^4+ \|X_1\|_{\Pcal,\infty}^2)^{\frac12}}{n}\EQ{|u''(Z^Q_n)-u''(\EQ{X_1})|^2}^{\frac12}\\
\leq \quad & 2L\frac{(4\|X_1\|_{\Pcal,\infty}^4+ \|X_1\|_{\Pcal,\infty}^2)^{\frac12}}{n}\EQ{|u'''(\zeta^Q)|^2|Z^Q_n-\EQ{X_1}|^2}^{\frac12}.
\end{align*}
Here $L$ is an upper bound of the continuous function $\op{int}\dom u\ni x\mapsto 1/u'(x)$ on the compact set $[\essinf_\Pcal X_1,\esssup_\Pcal X_1]$, and we have used that $z_n^Q\leq u(\EQ{X_1})$ (because $u''\leq 0$) while the function  $\op{int}\dom u\ni x\mapsto (u^{-1})'(x)=\frac{1}{u'(u^{-1}(x))}$ is  positive, increasing, and continuous, so $(u^{-1})'(z^Q_n)\leq 1/u'(\EQ{X_1})$. We also used the rough estimate $$M_Q(X_1)+\sigma_Q^2(X_1)\leq 16\|X_1\|_{\Pcal,\infty}^4+ 4 \|X_1\|_{\Pcal,\infty}^2.$$ Moreover, $\zeta^Q$ is a random variable taking values between $Z^Q_n$ and $\EQ{X_1}$. Note that $Z_n^Q$ is the random variable taking values between $\EQ{X_1}$ and $S_n/n$ corresponding to $Z_n$ in \eqref{eq3}. Estimating $|u'''(\zeta^Q)|$ from above by an upper bound $\hat L$ of the continuous function $\op{int}\dom u\ni x\mapsto |u'''(x)|$ on the compact set $[\essinf_\Pcal X_1,\esssup_\Pcal X_1]$, and using $|Z_n^Q-\EQ{X_1}|\leq |S_n/n-\EQ{X_1}|$ we conclude that
\begin{eqnarray*}\EQ{|u'''(\zeta^Q)|^2|Z^Q_n-\EQ{X_1}|^2}&\leq& \hat L \EQ{\left|\frac{S_n}{n}-\EQ{X_1}\right|^2}\\ &\leq &  \frac{\hat L}{n}\sigma^2_Q(X_1)\quad \leq\quad  \frac{4\|X_1\|_{\Pcal,\infty}^2 \hat L}{n}. \end{eqnarray*}
Thus, there is a constant $K>0$, depending on $X_1$ and $u$, but independent of $n$ and $Q$, such that
\begin{equation*}
\Delta(n,Q) \leq  \frac{K}{n^{3/2}}.
\end{equation*}

By similar arguments as above, we observe that the second term in \eqref{eq3} satisfies
\begin{align*}
\Gamma(n,Q) := & \left|(u^{-1})'(z_n^Q)-(u^{-1})'(u(\EQ{X_1}))\right|\frac{|u''(\EQ{X_1})|\sigma_Q^2(X_1)}{2n} \\\leq & \left|(u^{-1})'\left(u(\EQ{X_1})+\frac12\EQ{u''(Z^Q_n)\left(\frac{S_n}{n}-\EQ{X_1}\right)^2}  \right)-(u^{-1})'\left(u(\EQ{X_1})\right)\right|\frac{\tilde L}{n}
\end{align*}
for some constant $\tilde L$ only depending on $X_1$ and $u$.
In the first factor we used that $(u^{-1})'$ is increasing and the bounds for $z_n^Q$ given after \eqref{eq3}.
As $u^{-1}$ is twice continuously differentiable there is $\eta_n^Q\in [u(\EQ{X_1})+\frac12\EQ{u''(Z^Q_n)(\frac{S_n}{n}-\EQ{X_1})^2},u(\EQ{X_1}) ]$ such that
\begin{align*}
 & \left|(u^{-1})'\left(u(\EQ{X_1})+\frac12\EQ{u''(Z^Q_n)\left(\frac{S_n}{n}-\EQ{X_1}\right)^2}  \right)-(u^{-1})'\left(u(\EQ{X_1})\right)\right|\\
 & \leq |(u^{-1})''(\eta^Q_n)|\frac12\EQ{|u''(Z^Q_n)|\left(\frac{S_n}{n}-\EQ{X_1}\right)^2}.
\end{align*}
Consider the estimate
\begin{equation*}\frac12\EQ{u''(Z^Q_n)\left(\frac{S_n}{n}-\EQ{X_1}\right)^2}\geq -\frac{\bar L}{2} \EQ{\left(\frac{S_n}{n}-\EQ{X_1}\right)^2}\geq -\bar L\frac{2\|X_1\|^2_{\Pcal,\infty}}{n}, \end{equation*}
where $-\bar L$ is a lower bound of $u''$ on the compact set $[\essinf_\Pcal X_1,\esssup_\Pcal X_1]$ (recall that $u''\leq 0$).
Thus, for $n_0\in \Nat$ large enough such that for all $n\geq n_0$, we have
$$\essinf_\Pcal X_1 -\bar L\frac{2\|X_1\|^2_{\Pcal,\infty}}{n} \in \op{int} \dom u,$$
and choosing an upper bound $\hat K$ of the continuous function $(u^{-1})''$ on the compact set $$[u(\essinf_\Pcal X_1 -\bar L\frac{2\|X_1\|^2_{\Pcal,\infty}}{n_0}),u(\esssup_\Pcal X_1)],$$ we may further estimate
\begin{eqnarray*}
|(u^{-1})''(\eta^Q_n)|\frac12\EQ{|u''(Z^Q_n)|\left(\frac{S_n}{n}-\EQ{X_1}\right)^2}  & \leq &  \hat K \bar L \frac{2\|X_1\|^2_{\Pcal,\infty}}{n}.
\end{eqnarray*}
Summing up, we have shown that there is a constant $\widetilde K$ depending on $X_1$ and $u$, but independent of $Q$ and $n$, such that
\begin{equation*}\Gamma(n,Q)\leq \frac{\widetilde K}{n^2}. \end{equation*}
Thus, using \eqref{eq:cont1} and our estimates above, we deduce that
\begin{eqnarray*}
n\left(\Ucal(X_1)-\Ucal_\Pcal\left(\frac{S_n}{n}\right)\right)&\leq& \frac{1}{n} + n\sup_{Q\in \Pcal}\left\{ \EQ{X_1} -U_{Q}\left(\frac{S_n}{n}\right)\right\}\\ &\leq & \frac12 \sup_{Q\in\Pcal}R(\EQ{X_1})\sigma^2_Q(X_1) + \frac{1+\widetilde K}{n}+ \frac{K}{\sqrt{n}},
\end{eqnarray*}
which proves \eqref{eq:robust:est1}.
As for \eqref{eq:robust:est2}, let $Q\in \Pcal$ such that $\Ucal(X_1)= \EQ{X_1}+\alpha(Q)$. Then,
\begin{equation*}\Ucal(X_1)-\Ucal_\Pcal\left(\frac{S_n}{n}\right)\geq \EQ{X_1}-U_Q\left(\frac{S_n}{n}\right).\end{equation*}
Hence, \eqref{eq:robust:est2} follows from Theorem~\ref{ce:thm1}.
\end{proof}

\begin{remark}
Note that the limit of $\Ucal_\Pcal(S_n/n)$, namely $\Ucal(X_1)$, depends on the penalty function $\alpha$,
but the speed of convergence derived in Theorem~\ref{prop:robust:est} does not;
it is the same for every $\alpha$ as long as $\Pcal$ and $u$ are left unchanged.
Moreover, the robustness induced by the class $\Pcal$ affects the speed of convergence only via worst case first moments and variances.
\end{remark}

\begin{remark}
If $\Pcal=\{P\}$, then Theorem~\ref{prop:robust:est} reduces to Theorem~\ref{ce:thm1}.
Hence, the bounds in Theorem~\ref{prop:robust:est} cannot in general be sharpened.
Moreover, under the conditions stated in Lemma~\ref{lem:robust1} we have that \eqref{eq:robust:est2} is always satisfied, so the lower bound is always sharp in that case.
Clearly, the upper bound \eqref{eq:robust:est1} may or may not be sharp in specific cases.
As already mentioned it is trivially sharp if $\Pcal=\{P\}$.
But suppose, for instance, that $X_i$ takes only two values, $1$ and $-1$, and that $\Pcal$ is rich enough in the sense that it contains probability measures which put all mass on any possible atom.
Moreover, suppose that $\alpha\equiv 0$.
Then, $\Ucal(X_1)=-1$ and also $\Ucal_\Pcal(S_n/n)=-1$, so the left hand side of \eqref{eq:robust:est1} equals $0$.
On the other hand, the right hand side of \eqref{eq:robust:est1} is larger than $\frac12 R(\mathrm{E}_Q[X_1])\sigma_Q^2(X_1)=\frac12 R(0)$ where $Q\in \Pcal$ is a probability measure such that $\mathrm{E}_Q[X_1]=0$.
Hence, in this case, provided $u''<0$, the upper bound is not sharp.
\end{remark}


\begin{example}\label{ex:speed}
Consider $u(x)=1-e^{-\gamma x}$ for some $\gamma>0$, so that the corresponding robust certainty equivalent is minus an entropy convex (coherent) risk measure.
Then (\ref{eq:robust:est1}) becomes
\begin{equation*}\limsup_{n\to \infty}\, n\left(\Ucal(X_1)-\Ucal_\Pcal\left(\frac{S_n}{n}\right)\right)\leq \frac12 \sup_{Q\in \Pcal}\gamma\sigma^2_Q(X_1).\end{equation*}
Moreover, if $Q\in \Pcal$ satisfies $\Ucal(X_1)=\EQ{X_1}+\alpha(Q)$
(see Lemma~\ref{lem:robust1}),
then (\ref{eq:robust:est2}) becomes
\begin{equation*}\liminf_{n\to \infty}\, n\left(\Ucal(X_1)-\Ucal_\Pcal\left(\frac{S_n}{n}\right)\right)\geq \frac12 \gamma\sigma^2_Q(X_1).\end{equation*}
Similarly, one can easily derive the bounds (\ref{eq:robust:est1}) and (\ref{eq:robust:est2}) explicitly for the power and log utilities
considered in Examples~\ref{ex:cert:equiv}.
\end{example}

\begin{examples}\label{ex:entr:ecov}(Risk sharing for entropy coherent and entropy convex measures of risk.)
Suppose that $n$ equally risk and ambiguity averse agents with entropy coherent measure of risk $\rho_{\mathcal{P},\gamma}$ at some level $\gamma>0$ pool their risks $X_1,\ldots,X_n$ and optimally relocate the aggregate risk $S_n=X_1+\ldots+X_n$.
In the face of model uncertainty, we assume that the random variables $X_1,\ldots,X_n$ belong to $L^\infty_\Pcal$ and are i.i.d.\ under any $Q\in\mathcal{P}$.
The optimal redistribution of risk is then given by $Y_i^*=S_n/n$, $i=1,\ldots,n$. 
Let $\gamma_n=\gamma/n$.
Then, similar to the non-robust case (see Example \ref{ex:entropic}),
\begin{equation}
\label{eq:gl2}
n\rho_{\Pcal,\gamma}(Y^*_i)=n\frac{1}{\gamma}\sup_{Q\in\mathcal{P}}\log \EQ{e^{-\gamma Y^*_i}} =\frac{1}{\gamma_n}\sup_{Q\in\mathcal{P}}\log \EQ{e^{-\gamma_n S_n}}=\rho_{\mathcal{P},\gamma_n}(S_n).
\end{equation}
Thus, 
per position,
$\rho_{\Pcal,\gamma}(Y^*_i)=\rho_{\Pcal,\gamma_n}(X_1)$.
According to Proposition~\ref{prop:robust:conv}, $\rho_{\Pcal,\gamma}(Y^*_i)\leq \rho_{\Pcal,\gamma}(X_1)$, i.e.,\
optimally pooling and relocating the aggregate $S_n$ is beneficial from the perspective of the entropy coherent measure of risk.
In this case, this is also easily verified directly, because $\gamma_n<\gamma$ (when $n\geq 2$) and Jensen's inequality imply
\begin{equation*}\rho_{\Pcal,\gamma}(Y^*_i)=\frac{1}{\gamma_{n}}\sup_{Q\in\mathcal{P}}\log \EQ{e^{-\gamma_{n}X_1}}\leq \frac{1}{\gamma}\sup_{Q\in\mathcal{P}}\log \EQ{e^{-\gamma X_1}}=\rho_{\mathcal{P},\gamma}(X_1).\end{equation*}
Moreover, by Proposition~\ref{prop:robust:conv}, we have that
\begin{equation*}\lim_{n\rightarrow\infty}\frac{1}{n}\frac{1}{\gamma_{n}}\sup_{Q\in\mathcal{P}}\log \EQ{e^{-\gamma_{n}S_n}}=\sup_{Q\in\mathcal{P}}\EQ{-X_1},\end{equation*}
with the speed of convergence given in Example~\ref{ex:speed}.

Assume now that the $n$ equally risk and ambiguity averse agents apply an entropy convex measure of risk $\rho_{\mathcal{P},\gamma,\alpha}$
at some level $\gamma>0$.
According to Lemma~\ref{lem:robust2}, 
the optimal redistribution is again equal to $Y_i^*=S_n/n$, and
\begin{equation*}
\sum_{i=1}^n\rho_{\mathcal{P},\gamma,\alpha}(Y_i)=n\sup_{Q\in\mathcal{P}}\{\rho_{Q,\gamma}(S_n/n)-\alpha(Q)\}=\sup_{Q\in\mathcal{P}}\{\rho_{Q,\gamma_n}(S_n)-n\alpha(Q)\}.
\end{equation*}
Note that this corresponds to an entropy convex risk measure at level $\gamma/n$ with penalty function $n\alpha$.
Moreover, per position,
$\tfrac{1}{n}\sup_{Q\in\mathcal{P}}\{\rho_{Q,\gamma/n}(S_n)-n\alpha(Q)\}=\sup_{Q\in\mathcal{Q}}\{\rho_{Q,\gamma_n}(X_1)-\alpha(Q)\}$,
and, by Proposition~\ref{prop:robust:conv}, we have that
\begin{equation*}\lim_{n\to \infty}\sup_{Q\in\mathcal{P}}\{\rho_{Q,\gamma_n}(X_1)-\alpha(Q)\} =\sup_{Q\in\mathcal{P}}\{\EQ{-X_1}-\alpha(Q)\},\end{equation*}
with the speed of convergence given in Example~\ref{ex:speed}.
\end{examples}

\begin{example}(Esscher densities and the relative entropy.) As is well-known (Csisz\'ar \cite{C75}),
we may represent the entropic risk measure $\rho_{P,\gamma}$ as a
robust expectation with a
supremum over the set of Esscher densities given by
\begin{equation}\label{eq:Esscher}\Pcal=\{Q\ll P\mid dQ/dP=e^{-\gamma X}/\EP{e^{-\gamma X}},\; X\in L^\infty(\Omega,{\cal F},P)\},\end{equation}
and with penalization $\alpha(Q)=\gamma H(Q|P)$, $Q\in \Pcal$, $\gamma>0$, where
\begin{equation*}
H(Q|P)=\left\{\begin{array}{cl}
&\mathrm{E}_{Q}\left[\log\Big(\dfrac{dQ}{dP}\Big)\right],\mbox{ if } Q\ll P;\\
& \infty, \mbox{ otherwise};
\end{array}
\right.
\end{equation*}
is the relative entropy or Kullback-Leibler divergence.
The relative entropy is a metric of the distance between the probability measures $Q$ and $P$
and a special case of a $\varphi$-divergence (Ben-Tal and Teboulle \cite{BT86,BT87}).
It is widely applied in macroeconomics (Hansen and Sargent \cite{HS01,HS07}),
decision theory (Strzalecki \cite{S11a,S11b}),
and financial mathematics (Frittelli \cite{F00} and F\"ollmer and Schied \cite{FS11}). 
Note that $\Pcal$ given in \eqref{eq:Esscher} is the minimal set needed to represent $\rho_{P,\gamma}$.

In the specific case of \eqref{eq:Esscher}, assuming that $X_i$ is i.i.d.\ for any $Q\in \Pcal$ is equivalent to $X_1=X_i$ being constant.
Indeed, to see this, note first that this $\Pcal$ contains all $Q\ll P$ such that $dQ/dP$ is bounded from above and bounded away from zero
(simply let $X=-\frac1\gamma\log \frac{dQ}{dP}$).
Next, suppose there is a Borel set $A$ of $\R$ such that $0<P(X_1\in A)<1$ and consider $Q$ given by $\frac{dQ}{dP}=\frac1c (21_{\{X_1\in A\}} + 1_{\{X_1\not\in A\}})$, where $c=2P(X_1\in A) + P(X_1\not\in A)$.
Then, as $P\in \Pcal$, we have that $X_1$ and $X_2$ are independent under $P$ and $P(X_1\in A)=P(X_2\in A)$.
At the same time, for $Q$ we obtain $Q(X_1\in A)=\frac2c P(X_1\in A)$ and
\begin{eqnarray*}Q(X_2\in A)&=&\frac1c ( 2P(X_1\in A)P(X_2\in A) + P(X_1\not\in A)P(X_2\in A)) \\ &=&\frac1c ( P(X_1\in A)P(X_2\in A) + P(X_2\in A))\\&=& \frac1c P(X_1\in A)(P(X_1\in A)+1) <\frac2cP(X_1\in A),\end{eqnarray*}
so $X_1$ and $X_2$ cannot be identically distributed under $Q$ which contradicts the existence of such a set $A$.
Hence, $X_i=X_1$ is constant.
Thus, due to the required i.i.d.\ assumption for all $Q\in \Pcal$,
the convergence results of this section induce degeneracy in the case of a robust certainty equivalent with (subjectively given) linear utility
and $\Pcal$ given by (\ref{eq:Esscher}).
However, recall that the results of Section~\ref{sec:cert:equiv} already apply to $\rho_{P,\gamma}$,
when viewed as minus the certainty equivalent under expected exponential utility, \eqref{eq:entropicce},
without inducing degeneracy; see Example~\ref{ex:entropic}.

Furthermore, invoking the Kolmogorov extension discussed in the beginning of this section,
one may, for a given random variable $X$ on $(\Sigma,{\cal A})$,
let $\widetilde \Pcal= \{P_\gamma\mid \gamma \geq 0\}$
where $P_\gamma$ is generated by the Esscher density $\frac{dP_\gamma}{dP}=\frac{e^{-\gamma X}}{{\mathrm E}_P[e^{-\gamma X}]}$
with respect to some reference probabilistic model $P$.

We finally note that, the use of the relative entropy as an example of a penalty function in the main results of this section
is not ruled out
despite its connection to (\ref{eq:Esscher}) (see e.g., F\"ollmer and Schied \cite{FS11}, Section 3.2);
just the set $\Pcal$ must be smaller than (\ref{eq:Esscher}).
\end{example}

\begin{remark}(Risk and ambiguity.)
Observe that while, as $n\rightarrow\infty$, \textit{risk} reduction by pooling is achieved and exploited to the full limit, in the sense that the difference between the robust expectation and the robust certainty equivalent, $\Ucal(X_1)-\Ucal_\Pcal(S_n/n)$, vanishes (just like the risk premium in the previous section in which $\Pcal=\{P\}$), \textit{ambiguity} still remains: the robust certainty equivalent $\Ucal_\Pcal(S_n/n)$ is normalized by the robust expectation $\Ucal(X_1)$ rather than by a plain expectation.

In the spirit of Marinacci~\cite{Marinacci02} (see also Epstein and Schneider~\cite{EpsteinSchneider07}, Section 2.2), the pool of cooperative agents may learn in the long run the true common loss distribution of the $X_{i}$'s from observations obtained by drawing from the loss distribution.
Indeed, adopting a setup in which agents can observe realizations from their common yet unknown loss distribution, one expects the agents to learn the true loss distribution as the number of observations tends to infinity.
This intuition is formalized in Marinacci~\cite{Marinacci02} in a predictive parametric setting by sampling with replacement from ambiguous urns.
Then ambiguity fades away when sampling frequently enough, and risk (unambiguous uncertainty) remains in the limit.
Our setting pertains naturally to the short(er) run perspective in which at most only a limited number of draws is observed which does not plausibly resolve ambiguity.
\end{remark}

\setcounter{equation}{0}

\section{Robust Risk Premia}\label{sec:rp:rob}

Rather than robustifying the certainty equivalent, as in Section \ref{sec:cert:rob},
one may consider to take ambiguity into account when defining the risk premium
via the equivalent utility equation.
Just like the robust certainty equivalents, the resulting \textit{robust risk premia} compound risk and ambiguity aversion.
We consider first the theory of homothetic preferences and next the variational preferences
as our model for decision under risk and ambiguity.

\subsection{The homothetic case}
Let $\Pcal$ be a non-empty set of probability measures on $(\Omega, {\cal F})$ (not necessarily dominated), as in Section~\ref{sec:cert:rob}.
Furthermore, let $\beta:\Pcal\to (0,\infty)$ be a penalty function (ambiguity index), with
\begin{equation}\label{eq:homoth:1}\inf_{Q\in \Pcal}\beta(Q)=1\quad \mbox{and}\quad \sup_{Q\in \Pcal}\beta(Q)\in \R.\end{equation}
The robust risk premium under homothetic preferences, $\pi(v,X)$,
with initial wealth $v\geq 0$ and $X\in L^\infty_\Pcal$ such that $\essinf_\Pcal X\in \dom u$, is obtained as the solution to
\begin{equation*}\inf_{Q\in \Pcal} \left\{ u(\EQ{v+X}\beta(Q)-\pi(v,X))\right\} = \inf_{Q\in \Pcal} \left\{ \EQ{u(v+X)}\beta(Q)\right\},\end{equation*}
or equivalently, by continuity and monotonicity of $u$,
\begin{equation*}u(\Wcal(v+X)-\pi(v,X)) = \inf_{Q\in \Pcal} \left\{ \EQ{u(v+X)}\beta(Q)\right\},\end{equation*} where
\begin{equation*}\Wcal(X):=\inf_{Q\in \Pcal} \EQ{X}\beta(Q), \quad X\in L^\infty_\Pcal.\end{equation*}
Here, the criterion
\begin{equation}\label{eq:homoth}\Wcal(u(X))=\inf_{Q\in \Pcal} \EQ{u(X)}\beta(Q), \quad X\in L^\infty_\Pcal,\end{equation}
corresponds to the homothetic preferences model as introduced by Cerreia-Vioglio et al. \cite{CMMM08} and Chateauneuf and Faro \cite{CF10};
see also Dana \cite{D05}.
Let \begin{equation*}
\Wcal_\Pcal(X):=u^{-1}\left(\inf_{Q\in \Pcal} \EQ{u(X)}\beta(Q)\right)=u^{-1}(\Wcal(u(X))),\quad X\in L^\infty_\Pcal.
\end{equation*}
Then,
\begin{equation}\pi(v,X)= \Wcal(v+X)-u^{-1}\left(\inf_{Q\in \Pcal} \EQ{u(v+X)}\beta(Q)\right)=\Wcal(v+X) - \Wcal_\Pcal(v+X).\end{equation}

Using similar arguments as in the proofs of Lemmas~\ref{lem:robust2} and \ref{lem:robust1}, one can prove the following results:

\begin{lemma} Consider $n$ agents with the same robust certainty equivalent criterion \eqref{eq:homoth}.
Let $W\in L^\infty_\Pcal$ be such that $W/n\in \operatorname{int} \dom u$ $\Pcal$-a.s.\ and $\essinf_\Pcal W/n\in \dom u$.
Moreover, suppose that $$\Wcal(u(W/n))=\min_{Q\in \Pcal}\EQ{u(W/n)}\beta(Q).$$
Then the allocation which assigns the share $W/n$ of $W$ to each agent is Pareto optimal.
\end{lemma}

\begin{lemma}\label{lem:homoth1} Suppose that $\Pcal$ is dominated by a probability measure $P$ on $(\Omega,{\cal F})$.
Moreover, suppose that $\Pcal$ is weakly-compact in the sense of Lemma~\ref{lem:robust1}, and that $\beta$ is weakly continuous, i.e.,\ $\sigma(L^1(\Omega,{\cal F},P),L^\infty(\Omega,{\cal F},P))$-continuous.
Then,
\begin{equation*}\Wcal(X)=\min_{Q\in \Pcal} \EQ{X}\beta(Q), \quad X\in L^\infty(\Omega,{\cal F},P).\end{equation*}
Moreover,
for all $X\in L^\infty(\Omega,{\cal F},P)$ with $\essinf_{\{P\}} X\in \dom u$, we have that
\begin{eqnarray*}\Wcal_\Pcal(X)&=&u^{-1}\left(\min_{Q\in \Pcal} \EQ{u(X)}\beta(Q)\right)\quad=\quad \min_{Q\in \Pcal} u^{-1}\left(\EQ{u(X)}\beta(Q)\right). \end{eqnarray*}
\end{lemma}
Clearly, if, in the dominated case, $\Pcal$ is weakly-compact and $\beta$ is weakly-continuous,
then \eqref{eq:homoth:1} is automatically satisfied modulo a multiplicative constant which we norm to 1.

Note that
\begin{equation*}\Wcal_\Pcal(X)= \inf_{Q\in \Pcal} u^{-1}\left({\mathrm E}_{\widetilde Q}\left[u(X)\right]\right),\end{equation*}
where $\widetilde Q$ is the measure on $(\Omega,{\cal F})$ given by $\beta(Q)Q$, and recall \eqref{eq:homoth:1}.
Using the property that $\beta$ is bounded, we observe that one may argue as in Proposition~\ref{prop:robust:conv} and Theorem~\ref{prop:robust:est} with $\alpha\equiv 0$ to find the following results.
To this end notice that $\Wcal(v+\frac{S_n}{n})=\Wcal(v+X_1)$.
(We omit the detailed proofs to save space.)

\begin{proposition}
Let $(X_i)_{i\in \Nat}\subset L^\infty_\Pcal$ be i.i.d.\ under all $Q\in \Pcal$, and suppose that  $\essinf_\Pcal X_1\in \op{int} \dom u$. Also let $S_n:=\sum_{i=1}^n X_i$.
Then $\Wcal_\Pcal(v+S_n/n)$ is increasing in $n$ with
\begin{equation*}\Wcal_\Pcal\left(v+\frac{S_n}{n}\right)\leq \Wcal(v+X_1),\end{equation*}
and
\begin{equation*}\lim_{n\to \infty}\Wcal_\Pcal\left(v+\frac{S_n}{n}\right)=\Wcal(v+X_1),
\qquad\text{i.e.,}\qquad
\lim_{n\to \infty}\pi\left(v,\frac{S_n}{n}\right)=0.\end{equation*}
Moreover, there is a constant $K$ depending on $X_1$ and $u$ such that
\begin{equation*}\limsup_{n\to \infty}\sqrt{n}\pi\left(v,\frac{S_n}{n}\right)=
\limsup_{n\to \infty}\, \sqrt{n}\left(\Wcal(v+X_1)-\Wcal_\Pcal\left(v+\frac{S_n}{n}\right)\right)\leq K.\end{equation*}
\end{proposition}

\begin{theorem}\label{thm:main:homoth}
Let $(X_i)_{i\in \Nat}\subset L^\infty_\Pcal$ be i.i.d.\ under all $Q\in \Pcal$, and suppose that $\essinf_\Pcal X_1\in \op{int}\dom u$. Also let $S_n:=\sum_{i=1}^n X_i$.
Furthermore, suppose that $u$ is three times continuously differentiable on $\op{int} \dom u$ whereas $u^{-1}$ is two times continuously differentiable on  $\op{int} \dom u^{-1}$. Then,
\begin{eqnarray*}\limsup_{n\to \infty}\, n\,\pi\left(v,\frac{S_n}{n}\right)&=&
\limsup_{n\to \infty}\, n\left(\Wcal(v+X_1)-\Wcal_\Pcal\left(v+\frac{S_n}{n}\right)\right)\\&\leq& \frac12 \sup_{Q\in \Pcal}R\left(\EQ{v+X_1}\beta(Q)\right)\sigma^2_Q(X_1)\beta(Q).\end{eqnarray*}
Moreover,  if $Q\in \Pcal$ satisfies $\Wcal(v+X_1)=\EQ{v+X_1}\beta(Q)$, then
\begin{equation*}\liminf_{n\to \infty}\, n\,\pi\left(v,\frac{S_n}{n}\right)\geq \frac12 R(\EQ{v+X_1}\beta(Q))\sigma^2_Q(X_1)\beta(Q).\end{equation*}
\end{theorem}

Hence, upon Pareto optimal pooling and relocation of risks,
the robust risk premium under homothetic preferences diminishes,
and eventually vanishes in the limit as the multitude of risks tends to infinity,
with a speed of convergence that can be controlled according to Theorem~\ref{thm:main:homoth}.

\begin{example}
Suppose that Lemma \ref{lem:homoth1} applies.
Consider the subfamily of the power utility family of the form $u(x)=x^{p}$, $x\geq 0$, where $0<p<1$,
(and $u(x)=-\infty$, $x< 0$),
normalized such that $u(0)=0$.
Then,
\begin{align*}\Wcal_\Pcal(X)&=\inf_{Q\in \Pcal} \left(\EQ{X^{p}}\beta(Q)\right)^{1/p}\\
&=\inf_{Q\in \Pcal}\tilde{\beta}(Q)\left(\EQ{X^{p}}\right)^{1/p}=\inf_{Q\in \Pcal}\tilde{\beta}(Q)\|X\|_{Q,p},
\end{align*}
with $\tilde{\beta}(Q)=\beta(Q)^{1/p}$ and $\|X\|_{Q,p}=\left(\EQ{X^{p}}\right)^{1/p}$.
That is, in this case, $\Wcal_\Pcal(X)$ is
a worst case ``$\tilde{\beta}(Q)$-accrued'' $p$-norm,
and exhibits the convergence $\limsup_{n\to \infty}\pi\left(v,\frac{S_n}{n}\right)=0$,
for which the convergence bounds in Theorem \ref{thm:main:homoth} apply.
\end{example}





\subsection{The variational case}
Let $\Pcal$ and $\alpha:\Pcal\to \R\cup\{\infty\}$ be as in Section~\ref{sec:cert:rob} and suppose,
for ease of exposition, that $\inf_{Q\in \Pcal}\alpha(Q)=0$.
We define the robust risk premium under variational preferences, $\pi(v,X)$, with initial wealth $v\geq 0$ and risk $X\in L^\infty_\Pcal$ with $\essinf_\Pcal X\in \dom u$, as the solution to
\begin{equation*}\inf_{Q\in \Pcal} \left\{ u(v+\EQ{X}+ \alpha(Q) -\pi(v,X))\right\} = \inf_{Q\in \Pcal} \left\{ \EQ{u(v+X)}+\alpha(Q)\right\},\end{equation*}
or equivalently,
\begin{equation*}u(v+\Ucal(X)-\pi(v,X)) = \inf_{Q\in \Pcal} \left\{ \EQ{u(v+X)}+\alpha(Q)\right\},\end{equation*}
where, as before,
$$\Ucal(X)=\inf_{Q\in \Pcal} \EQ{X}+\alpha(Q), \quad X\in L^\infty_\Pcal,$$
is the corresponding robust expectation.
Here, the criterion
\begin{equation}\label{eq:var}\Ccal(X):=\Ucal(u(X))=\inf_{Q\in \Pcal} \EQ{u(X)}+\alpha(Q), \quad X\in L^\infty_\Pcal,\end{equation}
corresponds to the variational preferences model as introduced by
Maccheroni, Marinacci and Rustichini \cite{MMR06}.
Let
\begin{equation*}
\Vcal_\Pcal(X):=u^{-1}\left(\inf_{Q\in \Pcal} \EQ{u(X)}+\alpha(Q)\right),\quad X\in L^\infty_\Pcal,
\end{equation*}
and
\begin{equation*}
\Vcal(X):=u^{-1}\left(\inf_{Q\in \Pcal} u(\EQ{X})+\alpha(Q)\right),\quad X\in L^\infty_\Pcal.
\end{equation*}
Then,
$$\pi(v,X)= v+\Ucal(X)-u^{-1}\left(\inf_{Q\in \Pcal} \EQ{u(v+X)}+\alpha(Q)\right)=v+\Ucal(X) - \Vcal_\Pcal(v+X).$$

By continuity and monotonicity of $u^{-1}$ we have for all $X\in L^\infty_\Pcal$ with $\essinf_\Pcal X\in \dom u$ that
\begin{eqnarray*}\Vcal_\Pcal(X)&=&u^{-1}\left(\inf_{Q\in \Pcal} \EQ{u(X)}+\alpha(Q)\right)\quad=\quad \inf_{Q\in \Pcal} u^{-1}\left(\EQ{u(X)}+\alpha(Q)\right),\end{eqnarray*}
and
\begin{eqnarray*}\Vcal(X)&=&u^{-1}\left(\inf_{Q\in \Pcal} u(\EQ{X})+\alpha(Q)\right)=\inf_{Q\in \Pcal} u^{-1}\left( u(\EQ{X})+\alpha(Q)\right).\end{eqnarray*}

Also in this case of robust risk premia under variational preferences,
the respective counterparts of Lemmas~\ref{lem:robust2} and \ref{lem:robust1} remain true:

\begin{lemma}\label{lem:rob:rp:1} Consider $n$ agents with the same variational criterion \eqref{eq:var}.
Let $W\in L^\infty_\Pcal$ be such that $W/n\in \operatorname{int} \dom u$ $P$-a.s.\ and $\essinf_\Pcal W/n\in \dom u$.
Moreover, suppose that
\begin{equation}\label{eq:new1}\Ccal(W/n)=\min_{Q\in \Pcal}\left(\EQ{u(W/n)}+\alpha(Q)\right).\end{equation}
Then the allocation which assigns the share $W/n$ of $W$ to each agent is Pareto optimal.
\end{lemma}



\begin{lemma}\label{lem:rob:rp:2} Suppose that $\Pcal$ is dominated by a probability measure $P$ on $(\Omega,{\cal F})$. Moreover, suppose that $\Pcal$ is weakly-compact, and that $\alpha$ is weakly lower semicontinuous. For all $X\in L^\infty_\Pcal$ with $\essinf_\Pcal X\in \dom u$ we have that
\begin{equation*}\Ccal(X)=\min_{Q\in \Pcal} \EQ{u(X)}+\alpha(Q),\end{equation*}
and
\begin{eqnarray*}\Vcal(X)&=&u^{-1}\left(\min_{Q\in \Pcal} u(\EQ{X})+\alpha(Q)\right)=\min_{Q\in \Pcal} u^{-1}\left( u(\EQ{X})+\alpha(Q)\right).\end{eqnarray*}
\end{lemma}
We omit the proofs of Lemmas~\ref{lem:rob:rp:1} and \ref{lem:rob:rp:2} to save space,
because they follow similar arguments as the proofs of Lemmas~\ref{lem:robust2} and \ref{lem:robust1}.

Furthermore, as a consequence of Proposition~\ref{prop:robust:conv} and Theorem~\ref{prop:robust:est},
we immediately obtain the following convergence result as a starting point
upon requiring $\alpha$ to be trivial---an assumption that we will drop later:

\begin{corollary}\label{coroll:robust:rp} Suppose that $\alpha(Q)=0$ for all $Q\in \Pcal$, then
\begin{equation*}\Vcal_\Pcal(X)=\Ucal_\Pcal(X) \quad \mbox{and}\quad \Vcal(X)=\Ucal(X),\end{equation*}
for all $X\in L^\infty_\Pcal$ with $\essinf_\Pcal X\in \dom u$.
Hence, under the conditions of Proposition~\ref{prop:robust:conv}, we have that
\begin{equation*}\lim_{n\to \infty} \pi\left(v,\frac{S_n}{n}\right)=0,\end{equation*}
and under the additional conditions stated in Theorem~\ref{prop:robust:est}, we have that
\begin{equation*}\limsup_{n\to \infty}\, n\pi\left(v,\frac{S_n}{n}\right)\leq \frac12 \sup_{Q\in \Pcal}R(v+\EQ{X_1})\sigma^2_Q(X_1),\end{equation*}
and, furthermore,
\begin{equation*}\liminf_{n\to \infty}\, n\pi\left(v,\frac{S_n}{n}\right)\geq \frac12 R(\EQ{v+X_1})\sigma^2_Q(X_1),\end{equation*}
whenever $Q\in \Pcal$ satisfies $\Vcal(X_1)=\EQ{X_1}$ (see Lemma~\ref{lem:rob:rp:2}).
\end{corollary}

For non-trivial $\alpha$, however, the robust risk premium under variational preferences, while still decreasing in the multitude of risks,
will not vanish in the limit:

\begin{proposition}\label{prop:rob:risk:var}
Let $(X_i)_{i\in \Nat}\subset L^\infty_\Pcal$ be i.i.d.\ under all $Q\in \Pcal$, and suppose that $\essinf_\Pcal X_1\in \op{int} \dom u$.
Also let $S_n:=\sum_{i=1}^n X_i$.
Then $\Vcal_\Pcal(S_n/n)$ is increasing in $n$ with
\begin{equation}\label{eq:robust:rp1}\lim_{n\to \infty}\Vcal_\Pcal(S_n/n)=\Vcal(X_1).\end{equation}
Hence,
\begin{equation*}\lim_{n\to \infty}\pi\left(v,\frac{S_n}{n}\right)=\Ucal(v+X_1)-\Vcal(v+X_1).\end{equation*}
\end{proposition}

\begin{proof}
The fact that $\Ccal(S_n/n)$ and thus also $\Vcal_\Pcal(S_n/n)$ is increasing in $n$ follows from \eqref{eq:conc:order}.
Jensen's inequality implies that
\begin{equation*}\Ccal(S_n/n)\leq \inf_{Q\in \Pcal}u(\EQ{X_1})+\alpha(Q),\end{equation*}
and thus also $\Vcal_\Pcal(S_n/n)\leq \Vcal(X_1)$.
Moreover, for appropriate $ Q_n \in \Pcal$, we have that
\begin{eqnarray*}
\inf_{Q\in \Pcal}\left\{u(\EQ{X_1})+\alpha(Q)\right\}-\Ccal\left(\frac{S_n}{n}\right) &\leq & u({\mathrm E}_{Q_n}\left[X_1\right])-{\mathrm E}_{Q_n}\left[u\left(\frac{S_n}{n}\right)\right] +\frac{1}{n^2}\\ &\leq &
{\mathrm E}_{Q_n}\left[u'(\xi_{Q_n})\left|\frac{S_n}{n}-{\mathrm E}_{Q_n}\left[X_1\right]\right|\right] +\frac{1}{n^2}\\ &\leq & L \frac{\sigma_{Q_n}(X_1)}{\sqrt{n}}+\frac{1}{n^2} \quad \leq \quad \frac{2L\|X_1\|_\infty}{\sqrt{n}} +\frac{1}{n^2},
\end{eqnarray*}
where $\xi_{Q_n}$ is a random variable taking values between ${\mathrm E}_{Q_n}\left[X_1\right]$ and $S_n/n$ and $L$ is an upper bound of $u'$ on the compact set $[\essinf_\Pcal X_1,\esssup_\Pcal X_1]$.
Hence, $\lim_{n\to \infty}\Vcal_\Pcal(S_n/n)=\Vcal(X_1)$ follows from continuity of $u^{-1}$.
\end{proof}

We can also prove a result on the rate of convergence in \eqref{eq:robust:rp1}, similar to Proposition~\ref{prop:robust:conv} and Theorem~\ref{prop:robust:est}.
However, if $\alpha$ is non-trivial, so that we are not in the case of Corollary~\ref{coroll:robust:rp},
the result will not look as simple as in Theorem~\ref{prop:robust:est}.
To see this, consider any $Q\in \Pcal$.
Then,
\begin{align}
& u^{-1}\left(u(\EQ{X_1})+\alpha(Q)\right)-u^{-1}\left(\EQ{u\left(\frac{S_n}{n}\right)}+\alpha(Q) \right) \nonumber \\
= \quad & (u^{-1})'(\xi)\left(u(\EQ{X_1})-\EQ{u\left(\frac{S_n}{n}\right)}\right) \nonumber \\ = \quad &
\frac{1}{u'(u^{-1}(\xi))}\EQ{-\frac12 u''(\eta)\left(\frac{S_n}{n}-\EQ{X_1}\right)^2},\label{eq:rp:new1}
\end{align}
where $\xi\in [u(\EQ{X_1})+\alpha(Q),\EQ{u(S_n/n)}+\alpha(Q)]$
and $\eta$ is a random variable taking values between $\EQ{X_1}$ and $S_n/n$.
As $n\to \infty$, we see that \begin{equation}\label{eq:rp:new2}n \eqref{eq:rp:new1}\; \to \; \frac{-u''(\EQ{X_1})}{u'(u^{-1}(u(\EQ{X_1})+\alpha(Q)))}\frac{\sigma^2_Q(X_1)}{2}. \end{equation}
Thus, only if $\alpha(Q)=0$ the first factor on the right hand side of \eqref{eq:rp:new2} equals $R(\EQ{X_1})$.

\begin{theorem}\label{thm:rob:risk:var}
Let $(X_i)_{i\in \Nat}\subset L^\infty_\Pcal$ be i.i.d.\ under all $Q\in \Pcal$, and suppose that $\essinf_\Pcal X_1\in \op{int} \dom u$.
Also let $S_n:=\sum_{i=1}^n X_i$.
Furthermore, suppose that $\sup_{Q\in \Pcal}\alpha(Q)<\infty$.
Then there exists a constant $K$ depending on $X_1$, $u$ and $\alpha$ such that
$$\limsup_{n\to \infty}\, n\left(\Vcal(X_1)-\Vcal_\Pcal\left(\frac{S_n}{n}\right)\right)\leq K .$$
Moreover, if $Q\in \Pcal$ satisfies $\Vcal(X_1)=u^{-1}\left(u(\EQ{X_1}+\alpha(Q))\right)$ (see Lemma~\ref{lem:rob:rp:1}), then
$$\liminf_{n\to \infty}\,  n \left(\Vcal(X_1)-\Vcal_\Pcal\left(\frac{S_n}{n}\right)\right) \geq \frac{-u''(\EQ{X_1})}{u'(u^{-1}(u(\EQ{X_1})+\alpha(Q)))}\frac{\sigma^2_Q(X_1)}{2}.$$
\end{theorem}

\begin{proof}
Let $n$ be large enough and $Q\in \Pcal$ (depending on $n$) be such that $$\Vcal_\Pcal(S_n/n)\geq u^{-1}\left(\EQ{u(S_n/n)}+\alpha(Q)-1/n^2\right).$$
As in \eqref{eq:rp:new1}, there is $\xi\in [u(\EQ{X_1})+\alpha(Q),\EQ{u(S_n/n)}+\alpha(Q)-1/n^2]$ and a random variable $\eta$ taking values between $\EQ{X_1}$ and $S_n/n$ such that
\begin{eqnarray*}\Vcal(X_1)-\Vcal_\Pcal\left(\frac{S_n}{n}\right)&\leq &\frac{1}{u'(u^{-1}(\xi))}\EQ{-u''(\eta)\left(\frac{S_n}{n}-\EQ{X_1}\right)^2}\\ &\leq & \frac{L}{u'(u^{-1}(u(\esssup X_1)+\hat L))}\frac{\sigma^2_Q(X_1)}{n} \\ &\leq & \frac{L}{u'(u^{-1}(u(\esssup X_1)+\hat L))}\frac{4\|X_1\|_\infty}{n},
\end{eqnarray*}
where $\hat L=\sup_{Q\in \Pcal}\alpha(Q)$ and where $L$ is an upper bound for $-u''$ on  $[\essinf_\Pcal X_1,\esssup_\Pcal X_1]$.
The final assertion follows from \eqref{eq:rp:new2}.
\end{proof}

Note that if $\Pcal$ is weakly compact and $\alpha$ is weakly continuous, then $\sup_{Q\in \Pcal}\alpha(Q) <\infty$ is automatically satisfied.


\begin{examples} Consider the robust risk premium under variational preferences of type
\begin{equation*}\Ccal(X)=\inf_{Q\in \Pcal}\EQ{1-e^{-\gamma X}},\end{equation*}
i.e.,\ $\alpha\equiv 0$ and $u$ is the exponential utility function.
Then we are in the situation of Corollary~\ref{coroll:robust:rp}.
Hence, the robust risk premium vanishes for $n\to \infty$ with the given speed of convergence.
Note that in this case $\Vcal_\Pcal(X)=\Ucal_\Pcal(X)=-\rho_{\Pcal,\gamma}(X)$ and $\Vcal(X)=\Ucal(X)=\inf_{Q\in \Pcal}\EQ{X}$ for all $X\in L^\infty_\Pcal$, so we are essentially back in the entropy coherent case considered in Examples~\ref{ex:entr:ecov}.

Now consider the robust risk premium under variational preferences of type
\begin{equation*}\Ccal(X)=\inf_{Q\in \Pcal}\EQ{1-e^{-\gamma X}} + \alpha(Q),\end{equation*}
where $\alpha$ is non-trivial.
In that case, Proposition~\ref{prop:rob:risk:var} implies that
\begin{equation}\label{eq:last}\lim_{n\to \infty}\pi\left(v,\frac{S_n}{n}\right)=v+\inf_{Q\in \Pcal}\left(\EQ{X_1}+\alpha(Q)\right)+\frac1\gamma \log\left(\inf_{Q\in \Pcal} e^{-\gamma(\EQ{X_1}+v)}+\alpha(Q)\right),\end{equation}
with the speed of convergence given in Theorem~\ref{thm:rob:risk:var}.
Note the difference to the entropy convex case in Examples~\ref{ex:entr:ecov}. 
\end{examples}

\section{Conclusions}\label{sec:con}

In this paper, we have derived the asymptotic behavior of the certainty equivalents and risk premia associated with the Pareto optimal risk sharing contract,
in an expanding pool of cooperative agents bearing a growing multitude of risks.
We have first studied the problem under classical expected utility preferences and next we have considered the more delicate case of
ambiguity averse preferences to develop a robust approach that explicitly takes uncertainty with respect to the probabilistic model into account.
Our results make explicit, in a general setting that allows for aversion to both risk and ambiguity, when and at what rate
the key principle of risk sharing by Pareto optimally pooling and relocating risks may be exploited to the full limit.
The results in this paper require the cooperating agents to have identical preferences.
In future work, one may analyze the same problem under heterogeneous preferences, in which case the Pareto optimality of the proportional risk sharing rule
will no longer remain valid.

\singlespacing

\section*{Acknowledgements}
We are very grateful to
Hans F\"ollmer for stimulating discussions. 
This research was funded in part by the Netherlands Organization for
Scientific Research (Laeven) under grant NWO VIDI-2009.

\end{document}